\documentclass{article}
\usepackage[utf8]{inputenc}
\usepackage{amsmath}
\usepackage{amssymb}
\usepackage{amsthm}
\usepackage{tikz}
\usepackage{graphicx}
\usepackage{float}
\usepackage{algorithm}
\usepackage[noend]{algpseudocode}
\usepackage{tabularx}
\usepackage[margin=1in]{geometry}
\usepackage{arydshln}
\usepackage{multirow}
\usepackage{wrapfig}
\usepackage{xcolor}
\usepackage{xspace}
\usepackage{hyperref}
\usepackage{comment}
\usepackage{booktabs}
\usepackage{todonotes}
\usetikzlibrary{positioning,calc,shapes,arrows}
\usetikzlibrary{backgrounds}
\usetikzlibrary{matrix,shadows,arrows} 
\usetikzlibrary{decorations.pathreplacing}

\def\etal.{et\penalty50\ al.}
\theoremstyle{plain}
\newtheorem{theorem}{Theorem}[section]
\newtheorem{lemma}[theorem]{Lemma}

\newtheorem{corollary}[theorem]{Corollary}

\theoremstyle{definition}

\theoremstyle{remark}

\theoremstyle{plain}
\newtheorem*{theorem*}{Theorem}

\DeclareMathOperator{\OPT}{OPT}

\newcommand{\RSiC}{$RSiC$\xspace}
\newcommand{\FF}{\textit{FirstFit}\xspace}
\newcommand{\BF}{\textit{BestFit}\xspace}
\newcommand{\NF}{\textit{NextFit}\xspace}
\newcommand{\WF}{\textit{WorstFit}\xspace}
\newcommand{\AF}{\textit{AnyFit}\xspace}

\newcommand*\samethanks[1][\value{footnote}]{\footnotemark[#1]}

\algnewcommand{\LineComment}[1]{\State \(\triangleright\) #1}

\providecommand{\keywords}[1]
{
  \small	
  \textbf{\textit{Keywords---}} #1
}

\title{Renting Servers in the Cloud: The Case of Equal Duration Jobs}
\author{Mahtab Masoori\\
Concordia University, CSSE\\
\href{mailto:mahtab.masoori@concordia.ca}{mahtab.masoori@concordia.ca\thanks{Research is supported by NSERC.}}\and
Lata Narayanan\\
Concordia University, CSSE\\
\href{mailto:lata.narayanan@concordia.ca}{lata.narayanan@concordia.ca}\samethanks \and
Denis Pankratov\\
Concordia University, CSSE\\
\href{mailto:denis.pankratov@concordia.ca}{denis.pankratov@concordia.ca}\samethanks}

\date{\today}

\begin{document}

\maketitle
\begin{abstract}
Renting servers in the cloud is a generalization of the bin packing problem, motivated by job allocation to servers in cloud computing applications. Jobs arrive in an online manner, and need to be assigned to servers; their duration and size are known at the time of arrival. There is an infinite supply of identical servers, each having one unit of computational capacity per unit of time. A server can be rented at any time and continues to be rented until all jobs assigned to it finish. The cost of an assignment is the sum of durations of rental periods of all servers. The goal is to assign jobs to servers to minimize the overall cost while satisfying server capacity constraints. We focus on analyzing two natural algorithms, \NF and \FF , for the case of jobs of equal duration. It is known that the competitive ratio of \NF and \FF are at most  $3$  and $4$ respectively for this case.
We prove a tight bound of $2$ on the competitive ratio of \NF.  For \FF, we establish a lower bound of $\approx 2.519$ on the competitive ratio, even when jobs have only two distinct arrival times 0 and $t$. Using the weight function technique, we show that this bound is almost tight when there are only two arrival times; we obtain an upper bound of $2.565$ on the asymptotic competitive ratio of \FF. In fact, we show  an upper bound of $\frac{168}{131}(1+t)$ on the asymptotic competitive ratio for any $t > 0.559$. For the case when jobs have arrival times $0$ and $1$ and duration $2$, we show a lower bound of $\approx 1.89$ and an upper bound of 2 on the strict competitive ratio of \FF. Finally, we show an upper bound of $3/2$ on the competitive ratio of long-running uniform servers.
\end{abstract}


\section{Introduction}
Consider a cloud service provider that makes resources available to its clients via a cloud-based system. Over time, various jobs are submitted to the system by multiple clients and must be assigned to the servers in order to satisfy their demanded resources. 
Each server has limited resources available and job assignment must respect these capacity constraints. The company's operating costs are proportional to the total time of servers being active. This is an online problem since jobs arrive in the system one after another, and the scheduler must decide how to assign these jobs to the servers to minimize overall costs. In our scenario migrating a job from one server to another after it has been assigned is prohibitively expensive, so all the assignment decisions are fully irrevocable. A natural question arises here: how should jobs be assigned to servers in order to obtain the minimum cost? Two natural strategies are \NF and \FF. In \NF, the scheduler keeps track of at most one ``open'' server and keeps assigning jobs to it while the server capacity can accommodate the demanded resources. Once a new job requests more resources than the available amount provided by the open server at that time, then the server is {\em closed} (meaning no future jobs will be assigned to it unless it is released), and a new server is opened to accommodate the job. In the \FF strategy, every active server is considered {\em open} when a job arrives. The servers are scanned in the order in which they were opened and the first server that can accommodate the job is assigned that job. If no such server exists a new server gets opened.

The above scenario is formalized in the so-called \textit{renting servers in the cloud} problem, denoted by \RSiC, which was recently introduced by Li \etal.~\cite{li2014dynamic}. There is an infinite supply of identical servers, each having capacity $1$. The job scheduler needs to process a sequence of jobs $\sigma=(\sigma_1, \ldots, \sigma_n)$, where job $i$ is a triple $\sigma_i = (x_i, s_i, f_i)$ such that $x_i \in (0,1]$ is the size of a job, $s_i$ is the arrival time of a job, and $f_i$ is the finishing time of a job, and assign each job to a server. The ordering of jobs satisfies $s_1 \le s_2 \le \cdots \le s_n$. The input sequence is chosen by the adversary subject to this restriction. A server can be rented starting at any time $t$ and it continues to be rented until all jobs assigned to it terminate. The cost of a single server is the duration of its rental period. The cost of a whole solution is the sum of costs of all rented servers.
The goal is to assign jobs to servers to minimize the overall cost while satisfying server capacity constraints, i.e., at every time $t$ the sum of all sizes of jobs assigned to a given server and active at time $t$ should not exceed $1$. Jobs must be processed in the adversarial order that they are given. Thus, this problem on an instance with $s_i = 0$ and $f_i = 1$ for all $i \in \{1, \ldots, n\}$ is equivalent to the classical bin packing problem on the corresponding instance $(x_1, \ldots, x_n)$. Clearly, \RSiC is a generalization of the bin packing problem, that introduces a temporal component to the classical problem. Indeed, \FF and \NF algorithms, as defined above, are direct analogues of the same algorithms for bin packing. The performance of online algorithms is measured by the notion of competitive ratio, which is the worst-case ratio (over the inputs) between the cost achieved by the online algorithm and the cost achieved by an offline optimal solution that knows the entire input instance in advance. One distinguishes two notions of competitive ratio -- {\em strict}, that is it applies to all inputs, and {\em asymptotic}, that is it applies only to large enough inputs ignoring $o(OPT)$ terms. This distinction becomes very important for the \RSiC problem. It is easy to see that \RSiC in its generality does not admit algorithms with a constant competitive ratio. Thus, the research so far has focused on the inputs parameterized by $\mu = \frac{\max_i (f_i -s_i)}{\min_i (f_i-s_i)}$, which is the ratio between the maximum and minimum duration of a job in the input.

The research literature on the classical bin packing problem is quite rich and spans over $50$ years of active research. The asymptotic competitive ratio of 1.7 for \FF was established in the early seventies~\cite{princeton1971performance}. It took over $40$ years to establish the same result for the strict competitive ratio~\cite{dosa2013first}. By now there are many different proofs of competitiveness of \FF and its variants.

The key technique that is used to establish the majority of upper bounds is called the {\em weighting technique}. The idea is to define a weight function $w$ which modifies sizes of items that satisfies the following two properties:
\begin{enumerate}
	\item for most of the bins $B$ of a given algorithm $ALG$ we have $\sum_{x \in B} w(x) \ge 1$,
	\item for any bin $B$ of an offline optimal solution $OPT$ we have $\sum_{x \in B} w(x) \le \rho$.
\end{enumerate}
If such a weight function $w$ is found for a given algorithm $ALG$, it implies a competitive ratio of at most $\rho$, since the sum of all weights of items is at least $ALG$ and at most $\rho \cdot OPT$. This technique can be seen to be an instance of the more general primal-dual framework applied to bin packing. Variations of this technique involve making the weights dependent on the schedules of $ALG$ and/or $OPT$, splitting the weight function into several functions, and applying amortization.

\RSiC is also related to the dynamic bin packing problem. In dynamic bin packing, the items arrive and depart from the system, and repacking of previously packed items may or may not be allowed depending on the version of the problem. The goal is to maintain a small number of bins at all times compared to an optimal solution. In \RSiC repacking is not allowed and the objective function is different. In fact, \RSiC aims to reduce the {\em total cost} of all rented servers. 
We describe the known results for the \RSiC problem in greater detail in Section~\ref{Section:Related-Work}. Due to the vast interest in cloud computing in recent years, this problem has received much attention from applied researchers and now there are attempts to prove theoretical guarantees. However, there is a significant gap between the known lower and upper bounds for this problem and for \NF and \FF, in particular. This is not surprising, in light of the history of analysis of \FF and its variants for the vanilla bin packing problem. Adding a temporal component makes the problem significantly more difficult to analyze. 

The known upper bound proofs on the performance of \FF in case of \RSiC are based on \emph{utilization}, denoted by $util(\sigma)$, and \emph{span}, denoted by $span(\sigma)$, of the input $\sigma$, defined as follows:
\[ util(\sigma) = \sum_{i=1}^n x_i (f_i-s_i) \text{ and } span(\sigma) = \left| \bigcup_{i=1}^n [s_i,f_i] \right|.\]
Thus, the utilization can be thought of as the total volume (size times duration) of all jobs, and span is the total duration of all jobs ignoring overlaps. The previous results, which are based on span and utilization, such as the upper bound $2\mu + 1$ for \NF and $\mu + 3$ for \FF, include additive terms $+1$ and $+3$ which arise from applying the span bound, and which we suspect to be too loose. One of the ways of making progress towards analyzing this problem is to reduce or even completely eliminate these additive terms. There is a natural reason for the appearance of these additive terms and it has to do with the distinction between the asymptotic and strict competitive ratios. Let $ALG(t)$ denote the number of servers that are used by $ALG$ at time $t$, and define $OPT(t)$ similarly. Suppose one proves that for all times $t$ we have $ALG(t) \le \rho \cdot OPT(t) + \delta$, which would be analogous to an ``asymptotic bound'' for a fixed time $t$, so to speak (note that this bound does not have to apply for each time $t$, and amortization tricks may be used to prove this bound on average, but we assume such a bound for each $t$ for simplicity of the argument). The overall cost of the algorithm is $ALG = \int_0^\infty ALG(t)\;dt \le \int_0^\infty \rho \cdot OPT(t) + \delta \;dt = \rho \cdot OPT + \delta \cdot span(\sigma) \le (\rho + \delta) OPT$. Thus, an ``asymptotic bound'' for a fixed time $t$ with competitive ratio $\rho$ gets converted into the overall bound of $(\rho + \delta)$ for the problem. The ``asymptotic bounds'' for a fixed time $t$ arise because often one needs to ignore one or two special servers to argue the desired inequality. Thus, one approach towards tightening these bounds is to attempt to prove ``strict bounds'' for a fixed time $t$, i.e., $ALG(t) \le \rho' OPT(t)$. This would avoid the appearance of the $span(\sigma)$ term. This is akin to the difference between proving the asymptotic competitive ratio of \FF and strict competitive ratio of \FF for the classical bin packing. We view this as the main obstacle to improving the known results.

We now describe one possible program to  overcome the above-mentioned obstacle.  First, generalize the weight function technique from the bin packing problem to the \RSiC problem to prove asymptotic bounds in restricted cases. Second, optimize these bounds and make them strict in restricted cases. Third, extend the bounds obtained in restricted cases to more general cases ultimately culminating in the tight analysis of general \RSiC.

In this paper we make progress towards the above program.  
We focus on the special case of $\mu = 1$, which we call equal duration. We also investigate a particular case of the \RSiC problem in which jobs have only two arrival times. Jobs having only two arrival times may seem like an artificial restriction; however, some real-life applications tend to have this property. For instance, in the Map-Reduce model~\cite{dean2008mapreduce}  the input is partitioned into roughly equal blocks that are processed in two phases: Map and Reduce, corresponding to input items arriving at two times in \RSiC model. During the map phase, the user's map function develops a set of key/value pairs in order to generate a set of intermediate key/value pairs. In the second phase, this set of intermediate key/value pairs will be passed to the user-written reduce function, which will merge the values related to the same intermediate key. Map-Reduce is widely used in cloud computing for automatic parallel processing and large-scale data distribution. 

Our main results for \RSiC for jobs of equal duration are summarized below: 
\begin{enumerate}
	\item We show a tight bound of $2$ on \NF; recall that the previous best bound was $2\mu+1=3$ for our case of equal duration jobs (Theorem~\ref{thm:ub_nf}).
	
	\item We demonstrate a lower bound of $\dfrac{34}{27}(1+t)$ on the competitive ratio for \FF where jobs have duration $1$ and arrival times $0$ and $t$ for $t \in (0,1)$ (Theorem~\ref{L.B.Theorem}). This improves upon the bin packing lower bound when $t > 0.35$. 

	\item Using the weight function technique, we show an upper bound of $\dfrac{168}{131}(1+t)$ on the asymptotic competitive ratio of \FF where jobs have duration $1$, arrival times $0$ and $t$ for $t \in [0.559,1)$ (Theorem~\ref{thm:ff_asymp_ub}).
 To the best of our knowledge this is the first time the weight function technique has been used in the analysis of \RSiC. 
When $t\rightarrow 1$, our results show that the competitive ratio of 
\FF lies between $2.519$ and $2.565$.

 \item We prove a strict competitive ratio of $2$ for \FF for the case when each job has duration $1$ and arrival times are either $0$ or  $1/2$ (Theorem~\ref{thm:2ub}). 
 
		\item We show a tight asymptotic bound of $2/3$ on the load of long-running uniform \FF servers when jobs have duration $2$ and integer arrival times $0, 1, 2, \ldots$ (Corollary~\ref{corollary:uniform_long_servers}). This implies a competitive ratio of $\approx 3/2$ for \FF in the long-running uniform case.

\end{enumerate}
Note that the last result demonstrates that if we consider the case of jobs of duration $2$ and integer arrival times $0, 1, 2, \ldots$ then those servers running for a long time automatically have a high load in the limit. This suggests that the difficult inputs for \FF are those where servers are short lived. 


The rest of this paper is organized as follows. The related work is summarized in Section~\ref{Section:Related-Work}. Section~\ref{Section:Preliminaries} introduces some basic notation and definition for the renting servers in the cloud problem. In Section~\ref{Section:NextFit}, we discuss the bounds for the \NF algorithm. Then, in Section~\ref{Section:FirstFit}, we discuss the bounds for the \FF algorithm. In Sections~\ref{sec:ff_lb} and ~\ref{sec:general_2_arrivals}, we present our lower and upper bounds for the \FF algorithm, respectively. The strict bound for \FF is analyzed in Section~\ref{Section:Strict}. We also discuss long running uniform \FF servers in Section~\ref{sec:long_servers}, which suggests that shorter-lived servers may result in worse competitive ratio. 
Finally, conclusions are made in Section~\ref{Section:Conclusion}. 
 
\section{Related Work}\label{Section:Related-Work}

The research literature related to bin packing and its variants is rather large. Google Scholar produces 500,000 records in response to ``bin packing'' query. We cannot do it justice in this short section, so instead we mention only the most relevant (to this paper) fundamental results, restricted to the online setting and mostly to the \FF and \NF algorithms.

The bin packing problem~\cite{lee1985simple} is one of the most fundamental \textit{NP-hard}~\cite{garey1979computers} problems in the field of algorithm design that has received a lot of attention among researchers. In the online setting of this problem, we are given a sequence of items with sizes in the range $(0,1]$. The goal is to pack these items in the minimum number of bins of capacity $1$. Many online algorithms have been designed for the bin packing problem; among them, \FF is among the most famous ones. This algorithm keeps the bins in the order in which they were opened and places a given item in the first bin that has enough space for it. If such a bin does not exist, it opens a new one.  Ullman in~\cite{princeton1971performance} established the asymptotic bound on the competitive ratio of \FF of the form $1.7 \cdot OPT +3$. The additive term $3$ was improved multiple times~\cite{garey1972worst, garey1976resource, xia2010tighter} until $2013$ when Dosa and Sgall demonstrated the absolute ratio $1.7$~\cite{dosa2013first}. Another algorithm that we investigate in this paper is related to the bin packing \NF algorithm. This algorithm only allows one bin to be open at a time. If there is space, it assigns a given item to the open bin; otherwise, it closes that bin and opens a new one. Johnson et al.~\cite{johnson1974worst} proved that the competitive ratio of \NF is $2$. 
There are other algorithms for this problem as well such as \BF with ratio $1.7$~\cite{princeton1971performance} and the $HARMONIC$ family of algorithms. The currently best known algorithm is from the latter family, and has competitive ratio of $1.57829$~\cite{balogh2017new}.
\FF and \BF belong to the more general family of \AF algorithms, which is essentially the family of greedy algorithms (never open a new bin unless you have to). In general, \AF has tight competitive ratio $2$; however, rather large sub-families of \AF have tight competitive ratio of $1.7$~\cite{johnson1973near}. However, in this paper, we do not go into detail about those algorithms.

The \emph{dynamic} bin packing, generalizing the vanilla bin packing, was first proposed by Coffman et al.~\cite{coffman1983dynamic}. In this problem, items with different sizes arrive and depart in sequential order. The size and arrival time of an item are revealed when the item arrives, but the duration of a job is hidden. Thus, the algorithm is notified of the departure of an item at the time the item departs from the system. The  objective here is to minimize the maximum number of bins used at any time during the execution. 
Coffman et al. consider the version where repacking of items is not allowed. They established a lower bound of $2.38$ for any online algorithm which was subsequently improved several times; the currently best known lower bound is   $2.66$~\cite{wong20128}.
 
 In the discrete version of the problem, each item has size $1/k$ for some integer $k$. With this restriction on the input, the competitive ratio of \BF and \WF were proved to be $3$ by Chan et al. in~\cite{chan2008dynamic}. In the same paper, the authors proved that \FF has a competitive ratio $2.4942$ for the discrete problem, which was later improved to $2.4842$ in~\cite{han2010dynamic}. However, no online algorithm can do better than $2.428$~\cite{chan2008dynamic}.

The \RSiC problem, which can be viewed as a new type of dynamic bin packing, was introduced by Li et al.~\cite{li2014dynamic} for the first time. In this problem, bins are referred to as servers, and the cost of each server is a function of the duration of its usage. Li et al. showed that no algorithm from the \AF family can achieve a competitive ratio less than $\mu$, which is the ratio of the maximum to the minimum duration of all items in the input sequence.
They showed that the competitive ratio of \FF is at most $2 \mu + 13$. Li et al. also analyzed \FF under restricted input families. More specifically, they proved that the competitive ratio of \FF  when the sizes of jobs are greater than or equal to $1/k$, where $k$ is a positive integer, is at most $k$, and when the sizes of jobs are less than or equal to $1/k$, then \FF has the competitive ratio at most $\frac{k}{k-1} \mu + \frac{6k}{k-1} + 1$. Li et al.~\cite{li2014dynamic} also presented the $Modified$ \FF algorithm which achieves  competitive ratio of at most $\dfrac{8}{7}\mu + \dfrac{55}{7}$ when $\mu$ is not known in advance and  $\mu + 8$ when  $\mu$ is known. Later, Kamali and Lopez-Ortiz~\cite{kamali2015efficient} presented a general lower bound for the \RSiC problem (not just for \AF algorithms) of  at least $ \dfrac{ \mu}{1+\varepsilon(\mu-1)}$ where $\varepsilon$ is a lower bound on the sizes of jobs. 
Kamali and Lopez-Ortiz also considered the \NF algorithm and proved that it has a competitive ratio of $2\mu + 1$ and introduced a variant of \NF with competitive ratio $\mu + 2$. Tang et al.~\cite{TangLRC2016} improved the analysis of \FF establishing the upper bound of $\mu+4$ on the competitive ratio. In his PhD thesis~\cite{ren2018combinatorial}, Ren improved the upper bound on the competitive ratio of \FF further to $\mu + 3$, which is the current state-of-the-art.
\section{Preliminaries}\label{Section:Preliminaries}

We use $[n]$ to denote the set $\{1, 2, \ldots, n\}$. The input to the general \RSiC problem is a sequence of $n$ items $\sigma = (\sigma_1, \ldots, \sigma_n)$, where item $i$ is a triple $\sigma_i = (x_i, s_i, f_i)$ denoting a job of size $x_i \in (0, 1]$ starting at time $s_i \in [0, \infty)$ and finishing at time $f_i > s_i$. The \emph{duration} of a job $i$ is $f_i - s_i$. We say that a job $\sigma_i = (x_i, s_i, f_i)$ is \emph{active} or \emph{alive} at time $t$ if $t \in [s_i, f_i)$. Two jobs $\sigma_i$ and $\sigma_j$ are said to be \emph{non-overlapping} if $[s_i, f_i) \cap [s_j, f_j) = \emptyset$, otherwise they are said to be \emph{overlapping}. Note that representing active times of jobs by half-open intervals allows us to consider the two jobs such that one starts at exactly the same time as the other one finishes as non-overlapping. The jobs are presented in an order that respects their starting times, i.e., $s_1 \le s_2 \le \cdots \le s_n$. For the same starting time $s$, jobs starting at time $s$ are presented in an adversarial order. An important parameter of the input sequence is $\mu(\sigma) := \frac{\max_i (f_i - s_i) }{\min_i (f_i-s_i)}$, which is the ratio of the maximum duration of a job to the minimum duration of a job in the sequence. The case of all jobs having \emph{equal duration} corresponds precisely to $\mu = 1$. In the case of jobs having equal duration $d$ (which shall be clear from the context) we specify each job $\sigma_i$ by its size and start time to simplify the notation, i.e., $\sigma_i=(x_i, s_i)$, since the finishing time $f_i$ is a function of the start time, i.e., $f_i = s_i + d$.

Next, we discuss some standard terminology that is used to describe a state of a server. \emph{Opening} a server refers to the beginning of the rent period of the server when the first job is assigned to it. A server is considered \emph{active} or \emph{alive} at time $t$ if it has at least one scheduled job that is active at $t$. A server is considered \emph{closed} when no future job will be assigned to it. Note that this is a property of an algorithm rather than the server itself, and that a server may remain active even after it is closed. When all jobs assigned to the server finish, the server is \emph{released}. Thus, a server is \emph{not active} after it is released.

In this paper we focus on the case of equal duration of jobs and two arrival times $0$ and $t \in (0, 1)$, where $t$ is arbitrary, and jobs of duration $1$. In some sections, we concentrate on arrival times $0$ and $1$ and jobs of duration $2$, unless stated otherwise; by scaling, this setting is equivalent to arrival times $0$ and $1/2$ and jobs of duration $1$.
Which of these notations is used in a particular section will be clear from the context. 

For $t \in [0, \infty)$ we use  $X(t)$ to denote the sum of all sizes of jobs that have start time $t$, i.e., 

\[X(t) = \sum_{i \in [n] : s_i =t} x_i\]

For $t_1 < t_2$ we use $X(t_1, t_2)$ to denote the sum of all sizes of jobs that have start time in the interval $(t_1, t_2]$:
\[X(t_1, t_2) = \sum_{i \in [n] : s_i \in \left(t_1,t_2\right]}x_i\]

For a server $B$ we use $x(B,t)$ to denote the sum of all sizes of jobs that are scheduled on $B$ and active at time $t$. The value $x(B,t)$ is also called the \emph{load} of server $B$ at time $t$.

For an algorithm $ALG$ we use $ALG(\sigma)$ to denote the value of the objective function achieved by $ALG$ on the input $\sigma$. For a specific time $t$ we use $ALG(\sigma, t)$ to denote the number of servers that are active (i.e., ``alive'') at time $t$. If the last finishing time of any job is $T$ then we have:
\[ ALG(\sigma) = \int_0^T ALG(\sigma, t) \; dt.\]
We omit $\sigma$ from the above notation, when it is clear from the context. We use similar notation for an optimal solution denoted by $OPT$.

When the input sequences result in all \FF servers having the same start time and the same finishing time (which, consequently, must be the largest finishing time of any job in the sequence), we refer to this scenario as the case of \emph{uniform servers}.

\section{Tight Competitive Ratio $2$ for \NF}\label{Section:NextFit}

In this section we prove a tight bound of $2$ on the competitive ratio of the \NF algorithm for the \RSiC problem when $\mu = 1$ and jobs have arbitrary start times. Recall that the \NF algorithm keeps at most one server open at any time, i.e., a server that can be assigned newly arriving jobs. If a given job does not fit in the open server, the algorithm closes the server (meaning it will not assign any new jobs to it from this point onward), opens a new server, and assigns the job to the newly opened server. Note that \NF never reopens a closed server. 
The lower bound of $2$ follows immediately from the lower bound on the competitive ratio of \NF for the bin packing problem and the fact that renting servers in the cloud problem generalizes bin packing. 
The upper bound of $2$ on the competitive ratio of \NF follows from the theorem below, which shows that the upper bound holds not only on the total cost produced by \NF, but, in fact, it holds at each individual time $t$.
\begin{theorem}\label{thm:ub_nf}
	Suppose that the duration of each job in $\sigma$ is exactly $1$ while arrival times $s_i$ are arbitrary. Then, for every time $t \in [0, \infty)$ we have:
	\[\NF(\sigma, t) \le 2 \cdot OPT(\sigma, t).\]
\end{theorem}
\begin{proof}
	If there are no active jobs at time $t$ then $\NF(t) = 0$ and the claim trivially holds.	If there are active jobs at time $t$ and $\NF(t) \le 2$ then we have $OPT(t) \ge 1$ and the claim also trivially holds.
	
	It remains to show the claim for those $t$ where $\NF(t) \ge 3$. Let $k = \NF(t)$ and let $B_1, \ldots, B_k$ be the servers of $\NF$ that are still active at time $t$ and ordered in the non-decreasing order of their opening times. Observe that since \NF maintains only one open server at a time, servers $B_1, \ldots, B_{k-1}$ must be closed at time $t$. We also use the following simple fact based on duration of each job being exactly $1$:  a job is active at time $t$ if and only if its arrival time is in $(t-1, t]$. Therefore we have:
	\begin{equation}
	\label{eq:total}
	X(t-1, t) = \sum_{i=1}^k x(B_i, t).
	\end{equation}
	Since $B_1$ is active at time $t$, it must have had a job scheduled on it in time interval $(t-1, t]$. Also, the job that closed $B_1$ must have arrived in $(t-1, t]$. We conclude that servers $B_2, B_3, \ldots, B_k$ were opened in $(t-1, t]$ and thus all jobs scheduled on these servers since their opening are still active at time $t$; see Figure~\ref{fig:NFfig}.
 \begin{figure}[H]
	\begin{center}
		\includegraphics[scale=1]{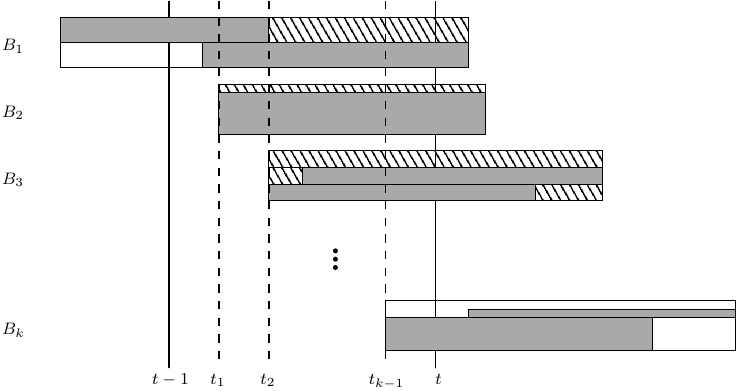}
		\caption{$B_1$ must have a job that arrived in $(t-1,t]$. Server $B_i$ was closed at time $t_i \in [t-1,t]$ by a job that was placed in $B_{i+1}$.}
		\label{fig:NFfig}
	\end{center}
\end{figure}

 Consider a server $B_i$ that got closed at time $t_i \in (t-1, t]$ by a job $\sigma_{z_i}$. This means that job $\sigma_{z_i}$ was scheduled on server $B_{i+1}$ so we have the inequality:
	\begin{equation*}
	x(B_i, t_i) + x(B_{i+1}, t_i) > 1 \;\;\; \text{for } i \in \{2, \ldots, k-1\}.
	\end{equation*}
	Combining this with the fact that all jobs scheduled on servers $B_2, \ldots, B_k$ are still active at time $t$, we conclude:
	\begin{equation}
	\label{eq:greater-one}
	x(B_i, t) + x(B_{i+1}, t) > 1 \;\;\; \text{for } i \in \{2, \ldots, k-1\}.
	\end{equation}
	Observe that the above inequality might fail for $i = 1$ since the jobs active at time $t_1$ in $B_1$ might finish before time $t$. Consider grouping servers in maximally many disjoint pairs $(2, 3)$, $(4, 5)$, $\ldots$, $(2j, 2j+1)$. The last pair must satisfy $2j+1 \le k$, which implies that $j = \lfloor \frac{k-1}{2} \rfloor$. Summing up equations~\eqref{eq:greater-one} corresponding to this pairing we derive:
	\begin{equation*}
	\sum_{q=1}^j (x(B_{2q}, t) + x(B_{2q+1}, t)) \geq j.
	\end{equation*}
	Combining this with equation~\eqref{eq:total} we conclude that $X(t-1, t) > j$. Rounding the left-hand side up we get $\lceil X(t-1, t) \rceil > j$, which implies that $\lceil X(t-1, t) \rceil \ge j+1$ by integrality. Multiplying both sides by $2$ we get: 
		\[2 \lceil X(t-1, t) \rceil \ge 2j + 2 = 2(j+1) = 2 \left( \lfloor \frac{k-1}{2} \rfloor+ 1\right) \ge 2 ( (k/2-1) +1) = k.\]
To conclude, we have demonstrated that $\NF(t) \le 2 \lceil X(t-1, t) \rceil$ and the equation $OPT(t) \ge \lceil X(t-1, t) \rceil$ is obvious.
\end{proof}
\begin{corollary}
	The competitive ratio of \NF for the case equal duration and arbitrary arrival times is at most $2$.
\end{corollary}


\section{\FF}\label{Section:FirstFit}
In this section we consider the \FF algorithm, which unlike \NF does not close servers unless they are no longer active. In the \FF algorithm, a newly arriving job is assigned to the earliest (in the order of opening) active server that can accommodate it. If none of the active servers has enough residual capacity to accommodate the new job, then a new server is opened. When all of the jobs in a server depart, the server is no longer active and it is closed. 

\subsection{Lower Bound}
\label{sec:ff_lb}
In this section, we present a parameterized lower bound on the competitive ratio of \FF for the case of equal duration and derive a couple of consequences of this parameterized lower bound. The adversarial instance consists of two sequences of jobs. In the first sequence all jobs have arrival time $0$ and in the second sequence all jobs have arrival time $t \in (0,1)$, by which our lower bound is parameterized. All jobs have duration $1$ in this instance. The first sequence is simply the nemesis instance for the bin packing problem due to~\cite{garey1972worst}. Recall that this instance creates three groups of \FF servers: (1) those with the load roughly $5/6$, (2) those with the load roughly $2/3$, and (3) those with the load roughly $1/2$. Meanwhile, $OPT$ is able to schedule all the jobs in the same instance into servers of duration $1$ and load close to $1$. Our second sequence consists of jobs that arrive at time $t$ and are used to extend the duration of each existing \FF server by an additive $1-t$. Thus, the items used to extend group (1) servers have size $1/12$ each. Those items used to extend group (2) have size $1/6$ each. Lastly, those items used to extend group (3) have size $1/4$ each. These items can be neatly combined by $OPT$ into servers of duration $1$ and load $1$. The following theorem presents the parameterized lower bound. The proof is self-contained, since we reproduce the nemesis instance from~\cite{garey1972worst} as part of the proof.
\begin{theorem}\label{L.B.Theorem}
The asymptotic competitive ratio of \FF for the \RSiC problem in the case where the duration of each job is $1$ and jobs arrive at times $0$ and $t \in (0,1)$ is at least $\dfrac{34}{27}(1+t)$.
\end{theorem}
\vspace{-1em}
\begin{proof} Fix a large positive integer $k$ and take $\delta$ such that $0 < \delta \ll 18^{-k}$. For $i \in [k]$ we define $\delta_i = \delta \times 18^{k - i}$. The adversarial sequence $\sigma$ of input items is a concatenation of two sequences. The jobs in the first sequence start at time zero, while the jobs in the second sequence start at time $t$. Furthermore, the first sequence is subdivided into three phases. The first phase consists of $k$ groups, where each one consists of $10$ jobs. Likewise, the second phase consists of $k$ groups, where each group consists of $10$ jobs. The third phase consists of $10k$ individual large jobs. Next, we describe the entire input with more details. 
	 
We begin by describing the first sequence.
In the first phase, group $i\in [k]$ consists of items $\sigma_{1,i}=(x_{1,i},s_{1,i}), \sigma_{2,i}=(x_{2,i},s_{2,i}), \ldots, \sigma_{10,i}=(x_{10,i},s_{10,i})$ where $s_{1,i}=s_{2,i}=\cdots=s_{10,i}=0$, and the sizes $x_{j,i}$ are defined as follows:
	
	\[x_{1,i} = \dfrac{1}{6} + 33\delta_i,
	x_{2,i} = \dfrac{1}{6} - 3\delta_i,
	x_{3,i} =  x_{4,i} = \dfrac{1}{6} - 7\delta_i,
	x_{5,i} = \dfrac{1}{6} - 13\delta_i,
	x_{6,i} = \dfrac{1}{6} + 9\delta_i,\]
	\[x_{7,i} = x_{8,i} = x_{9,i} = x_{10,i} = \dfrac{1}{6} - 2\delta_i.\]


	\FF puts $\sigma_{1,i}, \ldots, \sigma_{5,i}$ into server $2i-1$, which results in the server having the load $\dfrac{5}{6} + 3\delta_i $. None of the jobs $\sigma_{6,i}, \ldots, \sigma_{10,i}$ fit into this server. So, \FF assigns these jobs to server $2i$ with resulting load of $\dfrac{5}{6} + \delta_i$.
	Note that, none of the jobs in group $i$ will fit in any of the previous servers from group $i-1$. The smallest load of a server corresponding to group $i-1$ jobs, which is server $2i-2$, is $\dfrac{5}{6} + \delta_{i-1} = \dfrac{5}{6} + 18 \delta_i$, while the smallest job in group $i$ is $\sigma_{5,i}$. Therefore, $\sigma_{5,i}$ cannot fit into server $2i-2$. Thus, \FF has to open $2$ servers for each group $i$, where in total $2k$ servers are used for all of the jobs in the first phase.
	
	The second phase has $k$ groups, where each one consists of $10$ jobs $\sigma'_{1,i}=(x'_{1,i},s'_{1,i}), \sigma'_{2,i}=(x'_{2,i},s'_{2,i}), \ldots, \sigma'_{10,i}=(x'_{10,i},s'_{10,i})$. Each job has start time $0$ and their sizes are given in order as follows:
	
	\[x^{'}_{1,i} = \dfrac{1}{3} + 46\delta_i,
	x^{'}_{2,i} = \dfrac{1}{3} - 34\delta_i,
	x^{'}_{3,i} = x^{'}_{4,i} = \dfrac{1}{3} + 6\delta_i,
	x^{'}_{5,i} = \dfrac{1}{3} + 12\delta_i ,
	x^{'}_{6,i} = \dfrac{1}{3} - 10\delta_i,\]
	\[x^{'}_{7,i} = x^{'}_{8,i} = x^{'}_{9,i} = x^{'}_{10,i} = \dfrac{1}{3} +\delta_i.\]
		According to \FF, jobs $\sigma^{'}_{1,i}$ and $\sigma^{'}_{2,i}$ are assigned to one server with the total size of $\dfrac{2}{3} + 12\delta_i$; and, jobs $\sigma^{'}_{3,i}$ and $\sigma^{'}_{4,i}$ are packed into one server with total load of $\dfrac{2}{3} + 12\delta_i$. It is clear that, the smallest job from the remaining items in this group cannot fit in any of the open servers. Thus, \FF has to place each pair $\sigma^{'}_{5,i}$ and $\sigma^{'}_{6,i}$; $\sigma^{'}_{7,i}$ and $\sigma^{'}_{8,i}$; $\sigma^{'}_{9,i}$ and $\sigma^{'}_{10,i}$ into different servers, each with a total load of $\dfrac{2}{3} + 2\delta_i$. Note that the smallest job in group $i$, which is $\sigma^{'}_{2,i}$, cannot fit into any server from group $i - 1$, as the load of all the servers is at least $\dfrac{2}{3} + 2\delta_{i - 1} = \dfrac{2}{3} + 36\delta_i$.
	In other words, \FF has to open $5$ servers for assigning jobs in each group $i$, employing $5k$ servers for all the jobs in the second phase.
	
	In the third phase, $10k$ jobs of size $\dfrac{1}{2} + \delta$ each are presented having start time $0$. Since none of these jobs fit into the previously opened servers, \FF has to assign these jobs to new servers, resulting in opening $10 k$ servers in this phase. This finishes the description of the first sequence.
	
	In the second sequence, all jobs have start time $t$, and their sizes are as follows. The first $2k$ jobs are of size $\dfrac{1}{12}$ each, which are followed by $5k$ jobs of size $\dfrac{1}{6}$ each, followed by $10k$ jobs of size $\dfrac{1}{4}$ each. \FF assigns the first $2k$ jobs to the servers that were opened in the first phase (one job per server), the next $5k$ jobs to the servers that were opened in phase $2$ (one job per server), and the last $10k$ jobs to the servers that were opened in phase $3$ (one job per server). 
	
	In general, \FF uses $2k$ servers in the first phase, $5k$ bins in the second phase  and $10k$ servers in the third phase. 
	The second sequence has the effect of extending the duration of each server (the period during which the server is active) to $1+t$. Therefore, the cost of \FF is $17k (1+t)$. 

	The optimal offline algorithm, $OPT$, uses $10k$ servers for all the jobs from phase $3$ (i.e., $10k$ jobs of size $\dfrac{1}{2} + \delta$). The remaining space in these bins will be filled with the following combinations of items from phases $1$ and $2$:
	\begin{itemize}
		\item  for all jobs $j$; $j \geq 3$, in group $i$: $OPT$ combines $\sigma_{j, i}$ with $\sigma^{'}_{j, i}$.
		\item since $\sigma_{1, i}$ cannot be combined with $\sigma^{'}_{1,i}$, $OPT$ combines: 
		\begin{enumerate}
			\item $\sigma_{1, i}$ with $\sigma^{'}_{2, i}$ and, 
			\item $\sigma_{2, i}$ with $\sigma^{'}_{1, i+1}$.
		\end{enumerate}
	\end{itemize} 
	According to this packing scheme, two jobs $\sigma_{2, k}$ and $\sigma^{'}_{1, 1}$ remain unpacked. Therefore, $OPT$ uses one extra server for packing these two jobs. Thus, $OPT$ uses $10k + 1$ servers for duration $1$ to pack all jobs in the first $3$ phases, i.e., the first sequence. For the jobs in the second sequence, $OPT$ uses $\dfrac{k}{6} + \dfrac{5k}{6} + \dfrac{10k}{4}$ servers for jobs that are released at time $t$. In total, $OPT$ uses $\dfrac{162 k}{12} + 1 = \frac{27}{2} k + 1 $ servers, each of duration $1$. 
	
	As a result, the competitive ratio of \FF for \RSiC is at least $\frac{17k(1+t)}{27k/2 + 1} \rightarrow \frac{34}{27}(1+t)$ as $k \rightarrow \infty$.
\end{proof}
Substituting $t \rightarrow 1$ and $t=1/2$ in  Theorem~\ref{L.B.Theorem}, we obtain the following corollaries:
\begin{corollary}
	The competitive ratio of \FF for the case of equal duration and arbitrary arrival times is at least $2.\overline{518}$.
\end{corollary}
\begin{corollary}
	The competitive ratio of \FF for the case of all jobs having duration $2$ and arrival times $0$ and $1$ is at least $1.\overline{8}$.
\end{corollary}

\subsection{Upper Bound for Jobs with Two Arrival Times}\label{sec:general_2_arrivals}
In this section, we prove an upper bound for the competitive ratio of \FF for \RSiC when all jobs have duration of precisely one and arrival times of $0$ and $t $. Our result improves upon the previous best known upper bound of $\mu+3=4$. Our proof  extends the weight function technique that was previously successfully applied to the bin packing problem to the \RSiC problem. To the best of our knowledge, this is the first time the weight function technique has been used in the analysis of \RSiC. 

For the case of jobs of equal duration $1$ and arrival times $0$ and $t$, each \FF server falls in one of the following three categories:

\textbf{Category C:} starts at $0$ and ends at $1$,

\textbf{Category D:} starts at $0$ and ends at $1+t$,

\textbf{Category E:} starts at $t$ and ends at $1+t$.

In other words, servers in the category $C$ have some items that arrived at time $0$ and no items that arrived at time $t$. Servers in category $D$ have some items that arrived at time $0$ and some items that arrived at time $t$. Servers in category $E$ have some items that arrived at time $t$ and no items that arrived at time $0$.

 Let $C_1, C_2, \ldots, C_{k_1}$ be all the servers in category $C$ listed in the order of their opening times. Let $D_1, D_2, \ldots, D_{k_2}$ be all servers in category $D$ listed in the order of their opening times. Lastly, let $E_1, E_2, \ldots, E_{k_3}$ be all servers in category $E$ listed in the order of their opening times.


 Although we have ordered servers according to their opening times within each category, we will sometimes need to see how servers of different categories interact. Thus, we will make use of the following observations:

\begin{enumerate}
	\item $C_i$ was opened before $E_j$ for every $C_i$ and $E_j$ 
	\item $D_i$ was opened before $E_j$ For every $D_i$ and $E_j$ 
\end{enumerate}

Thus, if we were to order $C_i, D_j, E_k$ according to their opening times, then all the $C_i$ would appear before any of the $E_k$ and all the $D_j$ would appear before any of the $E_k$; however, $C_i$ and $D_j$ could be interleaved.

 Recall that for a server $B$ opened by \FF, $x(B, 0)$ denotes the sum of all sizes of jobs assigned to $B$ that were active at time $0$. Since the only jobs active at time $0$ are those with arrival time $0$, we equivalently can say that $x(B,0)$ is the sum of sizes of jobs assigned to $B$ with arrival time $0$. Similarly, note that at time $1$ the only jobs that are active in $B$ are those that arrived at time $t > 0$ and were assigned to $B$. Thus, $x(B, 1)$ denotes the sum of all sizes of jobs assigned to $B$ with arrival time $t$. The sum of all the sizes of jobs assigned to $B$ is called the total load on the server $B$ and is denoted by $X(B)$, that is, $X(B) = x(B,0) + x(B,1)$. We now prove lower bounds on loads of different categories of servers.

\begin{lemma}\label{lemma:category_C&E}
    For servers of category $C$ or $E$ used by \FF, at most one server can have a total load less than $1/2$, within each category. Moreover, if there are servers from both categories $C$ and $E$ used by \FF, then at most one server across both categories can have a load less than $1/2$ concurrently.
\end{lemma}
\begin{proof}
Suppose, by contradiction, that we have two servers $B_i$ and $B_j$ with $i<j$ of category $C$, that have load less than $1/2$. This implies that $x(B_i, 0) + x(B_j, 0) \le 1$. Therefore, \FF should not have opened server $B_j$ because its items could have been accommodated by server $B_i$, contradicting the assumption.
The same reasoning applies if both $B_i$ and $B_j$ belong to category $E$. Suppose we have a server $B_i$ of category $C$ and a server $B_j$ of category $E$, each with a load less than $1/2$. According to \FF, items in server $B_j$ should have been assigned to server $B_i$, since $x(B_i, 0) + x(B_j, t) \leq 1$. Therefore, \FF should not have opened server $B_j$, leading to a contradiction.
    
\end{proof}

\begin{lemma}\label{lemma:category_D1}
	For servers of category $D$ used by \FF, no more than two servers can have a total load less than $3/4$.
\end{lemma}

\begin{proof}
	Suppose, by contradiction, that we have three bins $B_1, B_2, B_3$ such that $X(B_i) < 3/4$ for $i = 1, 2, 3$. Without loss of generality, assume that the three bins were opened by \FF in the order $B_1, B_2, B_3$. Due to the definition of \FF, it follows that each item in $B_2$ and $B_3$ has size at least $1/4$ (otherwise, this item would have been placed into $B_1$). In particular, $x(B_2,1), x(B_3,1) > 1/4$. It follows that $x(B_2,0), x(B_3,0) < 3/4-1/4 = 1/2$. This is a contradiction since $x(B_3, 0)$ 
	should have been placed into bin $B_2$ by the \FF.
\end{proof}

In light of the previous lemma, by ignoring at most $2$ servers, we can assume that every server of category $D$ in \FF has total load at least $3/4$. Inspired by the lower bound construction from Section~\ref{sec:ff_lb}, we note that there are certain thresholds for $x(B,0)$ and $X(B)$ that are important for the proof. More specifically, $x(B, 0)$ thresholds are $5/6, 2/3, 1/2$; whereas $X(B)$ thresholds are $11/12, 5/6, 3/4$. For the sake of the analysis, we do not have to consider all pairs of thresholds. This follows from the following lemma:

\begin{lemma}\label{lemma:category_D2}
	Let $\alpha \in (0,1)$ and let $B_1, B_2, \ldots, B_k$ be the servers opened by \FF (in that order) that satisfy $x(B_i, 0) \ge \alpha$. Then at most one server $B_i$ satisfies $X(B_i) < (1+\alpha)/2$.
\end{lemma}
\begin{proof}
	Consider two servers $B_i, B_j$ with $i < j$ and assume, for contradiction, that $X(B_i), X(B_j) < (1+\alpha)/2$. This implies that $x(B_j, 1) = X(B_j)- x(B_j, 0) < (1+\alpha)/2 - \alpha = (1-\alpha)/2$, so the $x(B_j, 1)$ items can fit into $B_i$, since $X(B_i) + x(B_j, 1) < (1+\alpha)/2+(1-\alpha)/2=1$. By the \FF rule these items should not have been placed in $B_j$.
\end{proof}

Thus, by ignoring at most three more servers, we may assume that all servers with $x(B, 0) \ge 5/6$ have $X(B) \ge 11/12$, all servers with $x(B, 0) \ge 2/3$ have $X(B) \ge 5/6$, and all servers with $x(B, 0) \ge 1/2$ have $X(B) \ge 3/4$.

\begin{lemma}\label{lemma:2items_for_23}
For all but at most one server $B$ with $1/2 < x(B,0) \le 2/3$, it holds that $B$ contains a single job at time $0$.
\end{lemma}
\begin{proof}
    Suppose for contradiction that there are two servers $B_1$ and $B_2$ containing at least two jobs at time $0$ and satisfying $1/2 < x(B_i,0) \le 2/3$. Without loss of generality, assume that $B_2$ was opened later. Since $x(B_2,0) \le 2/3$, one of the jobs assigned to $B_2$ at time $0$ must have size at most $1/3$. Then this job should have been placed in $B_1$ by \FF.
\end{proof}

For the purpose of analysis we divide  servers in category $D$ into three types as shown in \autoref{tab:server-types}.
Observe that by the previous discussion, a Type I server also satisfies \textbf{$X(B) \ge 11/12$}. Type II server is further subdivided into two subtypes: Type II(a) satisfies $X(B) \ge 11/12$, Type II(b) satisfies $5/6 \le X(B) < 11/12$. 
Similarly Type III server is further subdivided into several subtypes: Type III(a) satisfies $X(B) \ge 11/12$, Type III(b) satisfies $5/6 \le X(B) < 11/12$, and Type III(c) satisfies $3/4 \le X(B) < 5/6$. By Lemma~\ref{lemma:2items_for_23} each Type III server contains exactly one item at time $0$. 

\begin{table}[!htb]
\caption{Different Types of  Servers in Category $D$}
\label{tab:server-types}
\centering
 \begin{tabular}{@{}l l ll c@{}} 
 \toprule
 Type & Subtype & Bounds on $x(B, 0)$ & Bounds on $X(B)$ & Num of items at time $0$\\ [0.5ex] 
 \midrule
 Type I & & $5/6 \le x(B, 0) \le 1$ & $11/12 \le X(B) \le 1$ & $\ge 1$\\ \addlinespace
 \hline \addlinespace
 Type II & &$2/3 \le x(B, 0) < 5/6$ & & \\
 & (a) &  & \texttt{$11/12 \le X(B) \le 1$} & $\ge 1$\\
 & (b) &  & $5/6 \le X(B) < 11/12$ & $\ge 1$\\ \addlinespace
 \hline \addlinespace
 Type III & & $1/2 < x(B, 0) < 2/3$ & &\\
  & (a) &  & $11/12 \le X(B) \le 1$ & $= 1$\\
 & (b) &  & $5/6 \le X(B) < 11/12$ & $=1$\\
 & (c) &  & $3/4 \le X(B) < 5/6$ & $=1$\\
 \bottomrule
\end{tabular}
\end{table}



\small
We make some further observations. Consider a server of Type II (b). We let $\epsilon_1, \epsilon_2 > 0$ be such that $x(B, 0) =5/6-\epsilon_1$ and $X(B) = 11/12-\epsilon_2$. It follows that $x(B, 1) = 1/12+\epsilon_1-\epsilon_2$. Also, since each job arriving at time $t$ (except possibly for the first server of each type) has to have size greater than $1/12$, it follows that $X(B) = x(B, 0) + x(B, 1) > 5/6-\epsilon_1 + 1/12 = 11/12 - \epsilon_1$. It follows that $\epsilon_1 > \epsilon_2$. In addition, we observe that $\epsilon_1 \le 1/6$ and $\epsilon_2 \le 1/12$.

Similar observations hold for a server of Type III (c). More specifically, we have $x(B, 0) = 2/3-\epsilon_1$ and $x(B, 1) = 1/6 + \epsilon_1 -\epsilon_2$. The choices of $\epsilon_1$ and $\epsilon_2$ depend on the server, of course; however, we always have $\epsilon_1 > \epsilon_2$ and $\epsilon_1 \le 1/6$ while $\epsilon_2 \le 1/12$.

With these notations, we are ready to establish an upper bound using the weighting technique. Define the following weight functions: $w_1$ for the items arriving at time $0$ and $w_2$ for the items arriving at time $t$:
    \[ w_1(x) = \frac{156+156t}{131} x + \left\{\begin{array}{ll} 0 & \text{ if } x \le 1/2 \\ \frac{12+12t}{131} &\text{otherwise} \end{array} \right. \hspace{1cm}\text{ and } \hspace{1cm} w_2(x) = \frac{168}{131} (1+t) x.\]
    By employing these two weight functions on the items packed into servers of categories $C, D$ and $E$ used by \FF, we can prove the following:


\begin{lemma} \label{theorem:bounds_for_servers}   
   On any input $\sigma$ with items arriving at two arrival times 0 and $t$, we have the following:
    \begin{itemize}
        \item[(i)] If $t \ge \frac{41}{90} \approx 0.456$ then $w_1(C_i) \ge 1 = d(C_i)$ for all but constantly many servers $C_i$ in category $C$.

        \item[(ii)] For all but constantly many servers $D_i$ in category $D$ it holds that $w_1(D_i)+ w_2(D_i) \ge 1 + t = d(D_i)$. 
        
        \item[(iii)] If $t \ge \frac{47}{84} \approx 0.560$ then $w_2(E_i) \ge 1 = d(E_i)$ for all but constantly many servers $E_i$ in category $E$. 

        \item[(iv)] If $t \ge \frac{1}{28} \approx 0.035 $ then $w_1(S) + w_2(S) \le \frac{168}{131}(1+t) d(S)$ for every server $S$ of $OPT$.
    \end{itemize}
\end{lemma}

\begin{proof}
\begin{itemize}
    \item[(i)] By Lemma~\ref{lemma:category_C&E} we may assume that  $x(C_i, 0) > 1/2$. If $x(C_i, 0) \le 2/3$ then by Lemma~\ref{lemma:2items_for_23}, we may assume that $C_i$ contains a single item. Therefore, we have:
    \[ w_1(C_i) \ge \frac{156+156 t}{131} x(C_i,0) + \frac{12+12t}{131} \ge \frac{156+156 (41/90)}{131} \cdot \frac{1}{2} + \frac{12+12(41/90)}{131} = 1,\]
where we used $t \ge (41/90)$ in the second inequality.

If $x(C_i, 0) > 2/3$, then we have:
\[w_1(C_i) \ge \frac{156+156 t}{131} x(C_i,0) \ge \frac{156+156 (41/90)}{131} \cdot \frac{2}{3} > 1.\]

    \item[(ii)] 	Let $D_i$ be a server in category $D$ that is used by \FF such that $X(D_i) \ge 11/12$. Then we have:
	\[ w_1(D_i) + w_2(D_i) \ge \frac{156+156 t}{131} X(D_i) \ge \frac{156}{131}(1+t) \frac{11}{12} > 1+t.\]

	Thus, the servers of Type I, Type II (a) and Type III (a) are handled in the sense that the weight function applied to each server of that type gives a total weight greater than $1+t$.
	
	We handle servers of Type II (b) next. Let $D_i$ be such a server and use the notation of $\epsilon_1$ and $\epsilon_2$ introduced prior to the theorem statement. In this case, we have
	\begin{align*}
	w_1(D_i) + w_2 (D_i)&=\frac{156}{131} (1+t) \left( \frac{5}{6} - \epsilon_1\right) + \frac{168}{131} (1+t) \left( \frac{1}{12} + \epsilon_1 - \epsilon_2\right) \\
	&= (1+t) \left( \frac{156}{131} \cdot \frac{5}{6} + \frac{168}{131} \cdot \frac{1}{12} + \frac{12}{131} \epsilon_1-\frac{168}{131} \epsilon_2  \right)\\
	&= (1+t) \left( \frac{130}{131}  + \frac{14}{131}  + \frac{12}{131} \epsilon_1-\frac{168}{131} \epsilon_2  \right) \\
	&>(1+t) \left( \frac{130}{131}  + \frac{14}{131}  + \frac{12}{131} \epsilon_2-\frac{168}{131} \epsilon_2  \right)\\
	&= (1+t) \left( \frac{130}{131}  + \frac{14}{131}  -\frac{156}{131} \epsilon_2  \right) \\
	&\ge (1+t) \left( \frac{130}{131}  + \frac{14}{131}  -\frac{13}{131} \right) = 1+t,
	\end{align*} 
	
	where the first inequality follows from $\epsilon_1 > \epsilon_2$ and the second from $\epsilon_2 \le 1/12$.
	

    Next, we handle Type III (b) servers. Let $D_i$ be such a server, and define $\epsilon$ to mean $x(D_i,0) = 1/2 + \epsilon$. Since $X(D_i) \ge 5/6$, it follows that $x(D_i, 1) \ge 5/6-(1/2+\epsilon) = 1/3 - \epsilon$. Also observe that $\epsilon < 1/6$. Plugging these estimates into the weight function we obtain:
    \begin{align*}
    w_1(D_i) + w_2(D_i) &= \frac{156}{131} (1+t) \left( \frac{1}{2} - \epsilon\right) + \frac{12}{131} (1+t) + \frac{168}{131} (1+t) \left( \frac{1}{3} - \epsilon\right)\\
    &= (1+t) \left( \frac{146}{131} - \frac{12}{131} \epsilon\right) > (1+t) \left( \frac{146}{131} - \frac{2}{131} \right) \ge 1+t,\\
    \end{align*}
    where we have used $\epsilon < 1/6$ in the penultimate step.
	
	We are only left to check the servers of Type III (c). Let $D_i$ be such a server and redefine $\epsilon_1$ and $\epsilon_2$ for this server. Then we have
	\vspace*{-0.1in}
	\begin{align*}
	w_1(D_i) + w_2 (D_i) &= \frac{156}{131} (1+t) \left( \frac{2}{3} - \epsilon_1\right) + \frac{12}{131} (1+t) + \frac{168}{131} (1+t) \left( \frac{1}{6} + \epsilon_1 - \epsilon_2\right) \\
	&= (1+t) \left( \frac{144}{131} + \frac{12}{131} \epsilon_1 - \frac{168}{131} \epsilon_2 \right) \ge (1+t),
	\end{align*} 
		\vspace*{-0.1in}
	
	Thus, we have proved that for all, but constantly $c$ many, servers of type $2$ used by \FF it holds that $w_1(D_i) + w_2(D_i) \ge 1+t$, considering $t > 0.3571$.

    \item[(iii)] Let $E_i$ be a server that starts at time $t$ and finishes at time $t+1$ which is of category $E$ of servers. By Lemma~\ref{lemma:category_C&E} we can assume this server has total load $> 1/2$. Therefore, when we apply the weight function $w_2$, we get: 
   \[ w_2(E_i) \ge \frac{168}{131} (1+t) \cdot \dfrac{1}{2} \ge  \frac{168(1+(47/84))}{262} = 1.\]

    \item[(iv)] Let $S$ be an arbitrary server of $OPT$. If $S$ contains only jobs that arrived at time $0$ then $w_1(S) \le \frac{156}{131}(1+t) + \frac{12}{131} (1+t) = \frac{168}{131}(1+t) = \frac{168}{131}(1+t) d(S)$, since $d(S) = 1$ in this case. Similarly, if $S$ contains only jobs that arrived at time $t$, then $w_2(S) \le \frac{168}{131}(1+t) = \frac{168}{131}(1+t) d(S)$, since again $d(S) = 1$. If $S$ contains both jobs arriving at time $0$ and time $t$ then the weight function is maximized when there is one item at time $0$ that barely exceeds $1/2$ (so that we collect the bonus for $w_1$) and the remaining items arrive at time $t$ and add up to barely $1/2$ (since $w_2$ point-wise is at least $w_1$). Thus, we have:
\begin{align*}
w_1(S) + w_2(S) &\le \frac{1}{2}\cdot  \frac{156}{131}(1+t)+ \frac{12}{131}(1+t) + \frac{1}{2}\cdot  \frac{168}{131} (1+t) \\
&= \frac{174}{131} (1+t) \le \frac{168}{131} (1+t) \cdot (1+t) \\
&= \frac{168}{131} (1+t) d(S),
\end{align*}
where the last inequality follows from $t \ge \frac{1}{28}$, and the last equality follows from $d(S) = 1+t$ in this case. Thus, in all cases, we show that $w_1(S) + w_2(S) \le \frac{168}{131} (1+t) d(S)$, where $d(S)$ denotes the duration of the server $S$.
\end{itemize}

\end{proof}

\begin{theorem}
 On input $\sigma$ where each job has duration $1$ and arrival time $0$ and $t > \frac{47}{84} \approx 0.560$, \FF achieves competitive ratio at most $\dfrac{168}{131}(1+t)$.
\end{theorem}\label{theorem:general_servers}

\begin{proof}

Let $x_1, \ldots, x_n$ denote the sizes of all jobs arriving at time $0$ and $x'_1, \ldots, x'_m$ denote sizes of all jobs arriving at time $t$. 
Let $B_1, \ldots, B_k$ denote \FF servers and $S_1, \ldots, S_p$ denote servers of $OPT$. Since $t > \frac{47}{84}$, parts (i)-(iii) of Lemma~\ref{theorem:bounds_for_servers} imply that $w_1(B_i) + w_2(B_i) \ge d(B_i)$ for all but some constant number $c$ of servers $B_i$. Similarly, part (iv) of Lemma~\ref{theorem:bounds_for_servers} implies that $w_1(S_i) + w_2(S_i) \le \frac{168}{131} (1+t)  d(S_i)$. Combining everything together we have:
\begin{align*} 
\FF-(1+t) c &\leq  \sum_{i=1}^k d(B_i) \le \sum_{i=1}^k (w_1(B_i) + w_2(B_i)) = \sum_{i=1}^n w_1(x_i) + \sum_{i=1}^m w_2(x'_i) \\&= \sum_{i=1}^p (w_1(S_i) + w_2(S_i))
\le \frac{168}{131} (1+t) \sum_{i=1}^p d(S_i) = \frac{168}{131} (1+t) OPT.
\end{align*} 
\end{proof}

By taking $t \rightarrow 1$, we obtain the following corollary:

\begin{corollary}
The asymptotic competitive ratio of \FF for the case of jobs of equal duration, two arbitrary arrival times $0$ and $t \in (\frac{47}{84} , 1]$ is at most $2.565$.
\end{corollary}

Examining various conditions on the range of $t$ in Lemma~\ref{theorem:bounds_for_servers}, we extend the range of $t$ where Theorem~\ref{theorem:general_servers} holds for inputs where \FF has specified behaviors.

\begin{corollary}\label{thm:ff_asymp_ub}
On inputs $\sigma$ where each job has duration $1$, arrival time $0$ and $t \in (\frac{41}{90},1)$, and no new servers are opened by \FF at time $t$, \FF achieves competitive ratio at most $\frac{168}{131} (1+t).$  
\end{corollary}

\begin{corollary}\label{thm:ff_asymp_ub}
On inputs $\sigma$ where each job has duration $1$, arrival time $0$ and $t \in [\frac{1}{28} . 1)$, and all \FF servers have  duration $1+t$, \FF achieves competitive ratio at most $\frac{168}{131} (1+t).$  
\end{corollary}

\subsection{Strict Upper Bound: Equal Duration $2$ and Arrival Times $0$ and $1$}\label{Section:Strict}
In this section, 
we prove an upper bound of $2$ on the  {\em strict} competitive ratio of \FF for \RSiC when all jobs have duration of precisely two and arrival times of $0$ and $1$. Clearly this is equivalent to the situation when all jobs have duration one, and arrival times of $0$ or $1/2$. The main result of this section is stated in Theorem~\ref{thm:2ub}. The proof is done by a careful case analysis. Since there are two arrival times, we have a lot more cases to deal with than in the  bin packing problem. We begin by setting up some notation that will be common to all lemmas in this section.

For the case of jobs of equal duration $2$ and arrival times $0$ and $1$, we have the same categories of servers as mentioned in Section~\ref{sec:general_2_arrivals}.
Recall that $x(B, j)$ denotes the total size of all items that are \emph{active} at time $j$ and were assigned to server $B$. If we wish to refer to items that arrived to server $B$ at time $1$, we can use the notation $x(B,2)$ since the only jobs active at time $2$ in $B$ are those that arrived at time $1$ (jobs that arrived at time $0$ are active during the half-open interval $[0,2)$).

We begin with a number of observations concerning the relationships between pairs of servers of the same type. 

\begin{lemma}
The following inequalities hold:
	\begin{align}
 x(C_i, 0) + x(C_{i+1}, 0) &> 1 \hspace{0.2cm} \mbox{ for }\hspace{0.2cm} i \in [k_1-1]\label{in:ci-chain} \\ 
 x(D_i, 0) + x(D_{i+1}, 0) &> 1 \hspace{0.2cm} \mbox{ for }\hspace{0.2cm} i \in [k_2-1]\label{in:di-chain}\\  	x(E_i, 2) + x(E_{i+1}, 2) &> 1 \hspace{0.2cm} \mbox{ for }\hspace{0.2cm} i \in [k_3-1]\label{in:ei-chain}   \\ 
 x(C_1, 0) + x(C_{k_1}, 0) &> 1 \hspace{0.2cm} \mbox{ if }\hspace{0.2cm} k_1 > 1
 \label{in:Cfirst-and-last}\\
  x(D_1, 0) + x(D_{k_2}, 0) &> 1 \hspace{0.2cm} \mbox{ if }\hspace{0.2cm} k_2 > 1 
 \label{in:Dfirst-and-last}\\
  x(E_1, 0) + x(E_{k_3}, 0) &> 1 \hspace{0.2cm}  \mbox{ if } \hspace{0.2cm} k_3 > 1 
 \label{in:Efirst-and-last} \\
x(D_i, 0) + x(D_i, 2) + x(D_{i+1}, 2) &> 1 \hspace{0.2cm} \mbox{ for } \hspace{0.2cm} i \in [k_2-1] \label{in:di-chain2} 
 \end{align}
 \end{lemma}
\begin{proof}
All the above inequalities follow directly from the definition of \FF and the ordering of the servers. For example, $C_{i+1}$ was opened by \FF because the first item to be placed in it did not fit in $C_i$, yielding Inequality (\ref{in:ci-chain}).  The same logic gives rise to a chain of inequalities (\ref{in:di-chain}) and (\ref{in:ei-chain}) concerning the $D_i$ and $E_i$ servers respectively, as well as (\ref{in:Cfirst-and-last}), (\ref{in:Dfirst-and-last}), and (\ref{in:Efirst-and-last}).
Note that the $x(D_{i+1}, 2)$ items didn't fit into $D_i$, which contained at most $x(D_i,0) + x(D_i, 2)$, hence the chain of inequalities (\ref{in:di-chain2}). 
\end{proof}

The next lemma concerns relationships between pairs of servers of {\em different} types.

\begin{lemma}
The following inequalities hold:
\begin{align}
 	x(D_{k_2}, 0) + x(D_{k_2}, 2) + x(E_{k_3}, 2) &> 1 
  \label{in:dlast-elast}\\
 	x(D_{k_2}, 0) + x(D_{k_2}, 2) + x(E_{1}, 2) &> 1 
  \label{in:dlast-efirst}\\  
  x(C_1, 0) + x(D_{k_2}, 0) &> 1  \hspace*{0.2in} 
  \label{in:cfirst-dlast}\\
   	x(C_i, 0) + x(E_{j}, 1) &> 1\hspace*{0.2in} \mbox { for  $i \in [k_3], j \in [k_2]$ }\label{in:c-and-e} 
    \end{align}
\end{lemma}

\begin{proof}
As observed earlier, every $E_j$ server was opened at time 1, after {\em all}  $C_i$ and $D_i$ servers were opened, and items were placed into an $E_j$ server because they did not fit into $D_i$ and $C_i$ servers, leading to Inequalities (\ref{in:dlast-efirst}),(\ref{in:dlast-elast}), and (\ref{in:c-and-e}. 
 Inequality (\ref{in:cfirst-dlast}) holds regardless of if $C_1$ was opened before or after $D_{k_2}$. 
\end{proof}

Define $A_1$ to be the sum of all items that arrived at time $0$ and were packed into $C_i$ servers, $A_2$ to be the sum of all items that arrived at time $0$ and were packed into $D_i$ servers, and $B_2$ to be the sum of all items that arrived at time $1$ and were packed into $D_i$ servers, and finally $B_3$ to be the sum of all items that arrived at time $1$ and were packed into $E_i$ servers. More specifically, we have:

\[ A_1 := \sum_{i=1}^{k_1} x(C_i, 0), \;\;\; A_2 := \sum_{i=1}^{k_2} x(D_i, 0), \;\;\; B_2 := \sum_{i=1}^{k_2} x(D_i, 2),\;\;\; B_3 :=\sum_{i=1}^{k_3} x(E_i, 2).\]
We also define $A := A_1 + A_2$ and $B = B_2 + B_3$

Since the duration of each item is 2, the following lower bound on the cost of OPT is immediate:

\begin{lemma} \label{lem:opt-lb}
$OPT \geq \lceil 2A + 2B \rceil$
\end{lemma}

Next we show some upper bounds on the number of different categories of servers, that hold when the number of servers of each category are large enough.

\begin{lemma}
The following inequalities hold:
\begin{eqnarray}
 	2 A_1 &>& k_1 \hspace*{0.2in}  \mbox{ if $k_1 > 1$} \label{eq:k1-ub1} \\
  	2 A_2 &>& k_2 \hspace*{0.2in}\mbox{ if $k_2 > 1$} \label{eq:k2-ub} \\ 
   2 B_3  & > & k_3    \hspace*{0.2in}\mbox{ if $k_3 >  1$} \label{eq:k3-ub} \\
   A_2 + 2B_2 & > & k_2 - 1 + x(D_{k_2}, 0) + x(D_1, 2) + x(D_{k_2}, 2)\hspace*{0.2in}\mbox{ if $k_2 > 1$}  \label{eq:k2-ub2} \\
   2A_1 + 2B_3 & > &  k_1 + k_3\hspace*{0.2in}\mbox{ if $k_1, k_3 \geq  1$} \label{eq:k1-and-k3-ub} \\
       4 A_2 + 4 B_2 & > & 3 k_2 -1 \hspace*{0.2in}
   \mbox{ if $k_2 > 2$}  \label{eq:k2-ub3}       
   \end{eqnarray}
\end{lemma}

\begin{proof}
  Adding up all the inequalities (\ref{in:ci-chain}) and (\ref{in:Cfirst-and-last}) gives us (\ref{eq:k1-ub1}). 
   Adding up all the inequalities (\ref{in:di-chain}) and (\ref{in:Dfirst-and-last}) gives us (\ref{eq:k2-ub});  adding up all the inequalities (\ref{in:ei-chain}) with (\ref{in:Efirst-and-last}) gives us (\ref{eq:k3-ub}), and adding all the inequalities (\ref{in:di-chain2}) gives us (\ref{eq:k2-ub2}).
   
   It follows from (\ref{in:c-and-e}) that $x(C_{k_1},0) + x(E_1, 2) > 1 $ and $x(C_1, 0) + x(E_{k_3}, 2)> 1$. Adding these two inequalities to the two chains of inequalities (\ref{in:ci-chain}) and (\ref{in:ei-chain}), we obtain (\ref{eq:k1-and-k3-ub}). (It can be verified that this is true even if one or both of $k_1, k_2$ equal 1.)



To show Inequality~\ref{eq:k2-ub3}, we consider the following $2$ cases:

	\textit{Case 1:}  $x(D_1,2) + x(D_{k_2},0) + x(D_{k_2},2) \ge 1/2 $: 
 Then we can write Inequality~\ref{eq:k2-ub2} as:
		$ A_2 + 2B_2 > k_2 - 1/2.$ Multiplying this inequality by two and adding to it Inequality~\ref{eq:k2-ub}, we obtain Inequality~\ref{eq:k2-ub3}.
 	
	\textit{Case 2:}  $x(D_1,2) + x(D_{k_2},0) + x(D_{k_2},2) < 1/2 $. Let us denote $x(D_{k_2},0)$ by $\epsilon$. 
		By \FF rules we $x(D_i, 0) > 1-\epsilon$ for $i \in \{1, \ldots, k_2-1\}$.
	
		Adding all these inequalities together with $x(D_{k_2},0)=\epsilon$ we obtain:
		\begin{equation}\label{case-2}
		A_2 > k_2 - 1 - (k_2 - 2)\epsilon. 
		\end{equation}
		From Inequality~\ref{eq:k2-ub2} it follows that:
		
		\begin{equation}\label{Case2}
		A_2 + 2B_2 > k_2 - 1 + \epsilon.
		\end{equation}
		
		If we multiply this inequality by $2$ and sum it up with twice the Inequality~\eqref{case-2} then we have:
		
		\begin{align*}
		4A_2 + 4B_2 &> 4(k_2 -1) - 2(k_2 - 2) \epsilon + 2\epsilon \\ 
		&>4(k_2 -1) - 2(k_2 - 3)\epsilon \\
		&>4(k_2- 1) - (k_2-3) \\
		& =4 k_2 - 4 - k_2 + 3  = 3 k_2 - 1,\label{Exception2}
		\end{align*}
		where the third inequality is due to $\epsilon < 1/2$. Observe that we silently assumed that $k_2 > 3$ in the above calculation. Observe that if $k_2 = 3$ we obtain $4A_2 + 4B > 4(k_2-1) - 2(k_2-3) \epsilon = 4 (k_2-1) = 4k_2-4$. However, for the case $k_2 = 3$ it holds that $4k_2-4 = 3 k_2 - 1$, so we conclude that the inequality $4A_2 + 4 B_2 \ge 3k_2-1$ holds for all $k_2 > 2$. 
\end{proof}

The following lemma relates the cost of $\FF$ to the cost of $\OPT$ and the number of category-$D$ servers used by $\FF$ by using a reduction to bin packing. 

\begin{lemma} \label{lem:ub-BP}
Let $\sigma$ be an input on which \FF uses $k_2$ servers of type $D$. Then
$$\FF(\sigma) \leq 1.7 \OPT(\sigma) + k_2 $$
\end{lemma}
\begin{proof}
Define $\sigma'$ to be an input derived from  $\sigma$ by shifting the arrival times of items arriving at time 1 to time 0, while preserving the original ordering of items. It is easy to see that for any solution for 
 $\sigma$, we can 'slide' the items arriving at time $1$ back to time $0$, thus obtaining a valid solution for $\sigma'$ without increasing the cost. Thus $OPT(\sigma') \le OPT(\sigma)$. Furthermore, we have  $\FF(\sigma') = \FF(\sigma)-k_2$, since all the $k_2$ type-D \FF servers have duration $2$ in the solution for $\sigma'$ instead of $3$ in the solution for $\sigma$. Lastly, since $\sigma'$ is just an instance of regular bin packing, we have $\FF(\sigma') \le 1.7 \cdot OPT(\sigma')$. Combining everything together, we get:
	$\FF(\sigma) = \FF(\sigma')+k_2 \le 1.7\cdot OPT(\sigma')+k_2
       \le 1.7\cdot OPT(\sigma)+k_2$
\end{proof}

Now, we are ready to state the main result of this section.

\begin{theorem}\label{thm:2ub}
	Let $\sigma$ be an input to the \RSiC problem. If the duration of each job in $\sigma$ is exactly $2$ and arrival times are $0$ and $1$ then $\FF(\sigma) \le 2 \cdot OPT(\sigma)$.
\end{theorem} 
\begin{proof}
We show that $\FF < 2 \cdot OPT + 1$ when \FF uses (a) exactly one type of server  (Lemma~\ref{lem:exactly-one-non-zero}) (b) exactly two types of servers	(Lemma~\ref{lem:k1=0}, ~\ref{lem:k2=0}, and ~\ref{lem:k3=0}) and (c) all three types of servers (Lemma~\ref{lem:general}). Then the result follows using integrality.
\end{proof}


We start with the case where only one type of server is used.

\begin{lemma} \label{lem:exactly-one-non-zero}
If exactly one of $k_1, k_2, k_3$ is non-zero, then $\FF <  2 \cdot OPT + 1$.
\end{lemma}

\begin{proof}
If $k_2 = 0$, we have $\FF \leq 1.7 \OPT$ by Lemma~\ref{lem:ub-BP}. So we assume below that $k_1=k_3=0$ and $k_2 \geq  1$. First note that if $k_2 = 1$, then \FF is optimal. If $k_2 = 2$ then the cost of \FF is $6$ (two servers of duration $3$). Meanwhile, $OPT$ has to open at least $2$ servers at time $0$; at time $1$, $OPT$ either opens a new bin with total cost $6$ or extends at least one server with total cost $\geq 5$. In both cases, the lemma follows. Therefore we assume $k_2 \ge 3$. We have $OPT \ge 2 (A_2 + B_2)$. Observing that $\FF = 3k_2$ and using Equation (\ref{eq:k2-ub3},  we conclude $ \FF < 2\cdot OPT + 1$.
\end{proof}

The next three lemmas consider the cases when exactly two types of servers are present. 

\begin{lemma} \label{lem:k1=0}
If $k_1=0$ and $k_2, k_3 \geq  1$, then $\FF < 2 \cdot OPT + 1$.
\end{lemma}

\begin{proof} We consider the following cases:
\begin{description}

\item[$k_2 = 1$]: Then there were items that arrived at time 0 as well as at time 1. At time 0, $\OPT$ had to open a server for a cost of 2. Since $k_3 \geq 1$, clearly the total size of items at time 1 exceeded 1, and $\OPT$ had to pay a minimum additional cost of 2. Thus  $\OPT \geq 4$, which implies that $1.7 \OPT + k_2 = 1.7 \OPT + 1 \leq 2 \OPT$. By Lemma~\ref{lem:ub-BP}, we have 
$$FF \leq 1.7 OPT + k_2 \leq 2 \OPT$$

\item[ $k_2 = 2$]: If $k_3 =1$, we have $\FF = 8$ and $\OPT \geq 5$, as at least two servers are needed at time 1 for a cost of 4, and $\OPT$ has to pay at least an additional cost of 1 for the items that arrive at time 1. Thus $\FF \leq 2 \OPT + 1$. Therefore let $k_3 \geq 2$. In this case, we claim $\OPT \geq  7$. To see this, observe that $\OPT$ must have opened at least two servers at time 0, and pay a cost of 2 between time 0 and 1. Next between time 1 and 2, since $x(D_1, 0)+ x(D_2, 0)> 1$ and $x(E_1, 2)+x(E_2,2)> 1$, and all these items are active in this time interval, $\OPT$ must pay a cost of at least 3. Finally between time 2 and 3, $\OPT$ must pay a cost of at least 2, for a total cost of at least 7.  We conclude that $1.7 \OPT + k_2 = 1.7 \OPT + 2 \leq 2 \OPT$. The result now follows by Lemma~\ref{lem:ub-BP}.



\item[$k_2 > 2$]:  
 If $k_3 > 1$,  we have $\FF = 3k_2 + 2 k_3$ and $OPT \geq 2(A_2 + B_2 + B_3)$. Therefore 
 $$\FF - 1 = 3k_2 - 1 + 2k_3 <  4 A_2 + 4B_2 + 4 B_3 \leq  2 \cdot OPT $$
 where the second inequality is derived by adding (\ref{eq:k2-ub3}) to twice (\ref{eq:k3-ub}).
	
	If instead $k_3 = 1$,  we have $\FF = 3k_2 + 2$.

 Then we add the inequalities  $x(D_2, 0) +x(D_3,0) > 1, \ldots, x(D_{k_2-1}, 0) + x(D_{k_2}, 0) > 1$ and $x(D_2, 0) + x(D_{k_2}, 0) > 1$ to  obtain:
	\begin{equation}\label{eq:a2exceptD1}
	2A_2 > k_2-1 + 2 x(D_1, 0).
	\end{equation}
	Adding (~\ref{eq:a2exceptD1}) to twice Inequality~\eqref{eq:k2-ub2} and to $4B_3 = 4 x(E_1, 2)$ we obtain:
	\begin{align*}
	4 A + 4 B &> 2k_2-2 + 2 x(D_1, 2) + 2 x(D_{k_2}, 0) + 2 x(D_{k_2}, 2)+ k_2-1 + 2 x(D_1, 0) + 4 x(E_1, 2) \\
	&= 3 k_2 -3 + 2(x(D_1, 0) + x(D_1, 2) + x(E_1, 2)) + 2 (x(D_{k_2}, 0) + x(D_{k_2}, 2) + x(E_1, 2)) \\
	&> 3k_2 + 1 = \FF-1,
	\end{align*}
	where the last inequality follows from $x(D_1, 0) + x(D_1, 2) + x(E_1, 2) > 1$ and $x(D_{k_2}, 0) + x(D_{k_2}, 2) + x(E_1, 2) > 1$. The  lemma follows from Lemma~\ref{lem:opt-lb}. 
\end{description}
\end{proof}

\begin{lemma}\label{lem:k2=0} 
	If $k_2=0$ and $k_1, k_3 \geq 1$ then $\FF < 2 \cdot OPT + 1$.
\end{lemma}
\begin{proof}
Follows from Lemma~\ref{lem:ub-BP}.
\end{proof}

\begin{lemma} \label{lem:k3=0}
If $k_3=0$ and $k_1, k_2 \geq 1$, then $\FF < 2 \cdot OPT + 1$.
\end{lemma}

\begin{proof}
 We consider the following cases: 
\begin{description}

\item[$k_2 =1$]: Observe that there must be items that arrived at time 0 as well as at time 1, since $k_2 \geq 1$.  \FF opens at least two servers at time 0, so it must be that $\OPT$ has to open at least two servers at this time, for a cost of 4. Then $\OPT$ has to pay at least an additional cost of 1 for the items arriving at time 1. Thus $\OPT \geq 5$.
The result now follows from Lemma~\ref{lem:ub-BP} and the fact that $1.7 \OPT + k_2 \leq 2 \OPT$ for $\OPT \geq 5$.

\item[$k_1= 1; k_2 > 1$]:  We have $\FF = 3k_2 + 2$. By adding  $x(D_{k_2}, 0) + x(C_1, 0) > 1$ to Inequality~\ref{eq:k2-ub2}, we obtain
\begin{equation}
A_1 + A_2 + 2B_2 > k_2
\end{equation}
If $A_1 > 1/2$, we have $2A_1>1$, and adding this to twice the above inequality and to Inequality~\ref{eq:k2-ub}, we get
$4A+4B > 3k_2+1$. Applying Lemma~\ref{lem:opt-lb} gives the desired bound. If instead $A_1 \leq 1/2$, we add the following series of inequalities:
		\begin{align*}
		x(D_i, 0)  &> 1 -A_1 \hspace{0.3cm} for \hspace{0.3 cm} i \in [1,..., k_2], and \\
		x(C_1, 0) &= A_1 
		\end{align*}
		
		to obtain:
		\begin{equation}\label{Eq:A_1=1}
		A_2 + A_1 > k_2 - A_1(k_2 - 1)    
		\end{equation}


		By adding $A_1 + x(D_{k_2}, 0) > 1 $ to Inequality~\eqref{Eq:A_1=1} and Inequality~\eqref{eq:k2-ub2} and multiplying by 2, we get:
		\begin{align*}
		4A_1 + 4A_2 + 4B &> 4 k_2 - 2A_1 (k_2 - 1)  + 2 x(D_1, 2) + 2 x(D_{k_2}, 2) \\
  & > 3k_2 + 1 + 2x(D_1, 2) + 2 x(D_{k_2}, 2) > 3k_2 + 1		\end{align*} 
  where the second inequality follows from $A_1 < 1/2$. Once again,
 Lemma~\ref{lem:opt-lb} gives the desired bound.

\item[$k_1, k_2 > 1$]:
	 We know that $ x(C_1, 0)+ x(C_{k_1}, 0) > 1$, therefore at least one of them is greater than $1/2$. If $x(C_{k_1}, 0) > 1/2$,  we add up Inequalities (\ref{in:ci-chain}), (\ref{in:di-chain2}) and (\ref{in:cfirst-dlast} to obtain: 
	\[ 2A_1 - x(C_{k_1},0) + A_2 + 2B  > k_1 + k_2 - 1 + x(D_1,2) + x(D_{k_2},2)\]
	Since $x(C_{k_1}, 0) > 1/2$ we have:
	\begin{equation} \label{eq:k_2+k_1-1/2}
	2 A_1 + A_2 + 2B > k_1 + k_2 - 1/2.
	\end{equation}
 The case $x(C_1, 0) > 1/2$ can be handled similarly to obtain (\ref{eq:k_2+k_1-1/2}), by  using $x(C_{k_1}, 0) + x(D_{k_2}, 0) > 1$ in place of (\ref{in:cfirst-dlast}).	
 
	Adding Inequality~\eqref{eq:k2-ub} and twice Inequality~\eqref{eq:k_2+k_1-1/2} we obtain
	$2 k_1 + 3 k_2 - 1 < 4 A + 4 B.$
	Since $2k_1 + 3k_2 = \FF$, using Lemma~\ref{lem:opt-lb},  we obtain
	$\FF < 2 \cdot OPT + 1$. 
 
\end{description}
\end{proof}

We end with the case when all three types of servers exist.

\begin{lemma} \label{lem:general}
If $k_1,k_2, k_3 \geq  1$, then $\FF <  2 \cdot OPT$ + 1.
\end{lemma}

\begin{proof}
We consider the following cases:

\begin{description}
\item[$k_2 =1$]: As in the proof of Lemma~\ref{lem:k3=0}, the lemma follows from Lemma~\ref{lem:ub-BP} and the  fact that $\OPT \geq 5$.

\item[$k_2 =2$]: If $k_1=k_3 = 1$ then $\FF = 10$ and $\OPT \geq 5$ as it has to open at least 2 servers at time 0 for a cost of 4, and pay at least an additional cost of 1 to accommodate items arriving at time 1. Therefore $\FF < 2 \OPT + 1$ in this case. Otherwise either $k_1 \geq 2$ or $k_2 \geq 2$. We claim that $\OPT \geq 7$ in both cases. If $k_1\geq 2$, then since $\FF$ opens at least 4 servers at time 0, it follows from the bin packing upper bound on \FF that $\OPT$ needs at least 3 servers at time 0 for a cost of 6; and then must pay at least an additional cost of 1 at time 1 to accommodate items arriving at time 1 for a total cost of at least 7. If $k_3 \geq  2$, then $OPT$ needs to pay at least cost $2$ between times 0 and 1, at least 3 between time 1 and 2, and at least 2 between time 2 and 3, for a total cost of at least 7. This completes the proof of the claim. 
Since $\OPT \geq 7$, by  Lemma~\ref{lem:ub-BP}, we have $\FF \leq 1.7 \OPT + 2 \leq 2 \OPT$ as needed.

\item[$k_2 > 2$]: Observe that $\FF = 2k_1 + 3k_2 + 2k_3$. Adding Inequality (\ref{eq:k2-ub3} to twice (\ref{eq:k1-and-k3-ub}), we obtain
\[ \FF - 1 = 2k_1 + 3k_2 + 2k_3 - 1 <  4A + 4B \leq  2\OPT \]

\end{description}
\end{proof}

\subsection{Upper Bound for Long-Running Uniform Servers}\label{sec:long_servers}

In this section we consider the case when all jobs in the input sequence $\sigma$ have duration $2$ and arrival times $0, 1, 2, \ldots, \ell$, and \FF packs these items into servers starting at time $0$ and finishing at time $\ell+2$. Thus, as we let $\ell$ go to infinity, this setting represents ``long-running'' uniform servers of \FF. We show that asymptotically (as $\ell \rightarrow \infty$) such \FF servers have amortized load of $2/3$ at all times. We also observe that this bound is tight, i.e., there are inputs on which long-running \FF servers have load of roughly $2/3$ at all times. Thus, long-running servers are beneficial for \FF since load of $2/3$ translates to competitive ratio of $3/2$ when \FF cost is compared to the cost of $OPT$. This suggests that worst-case adversarial instances for \FF on equal duration jobs should be such that \FF servers are short-lived.

We begin with a few basic observations about long running servers in the next lemma, before giving our main result in Lemma~\ref{lem:uniform_long_servers}.
\begin{lemma}\label{lemma:observations}
	Fix $\ell \in \mathbb{N}$. Let $\sigma$ be such that all jobs in $\sigma$ have duration $2$ and arrival times $0, 1, 2, \ldots, \ell$ and let $\mathcal{B}$ be a set of \FF servers of duration exactly $\ell+2$ that were opened at time $0$. 
	We denote the servers in $\mathcal{B}$ by $B_1, B_2, \ldots, B_k$ (opened in this order), where $k = |\mathcal{B}|$.
	
	For $i \in \mathbb{Z}_{\ge 0}$ we define the layer $i$, denoted by $L_i$, to be the set of all items packed in $\mathcal{B}$ that arrived at time $i$. Let $x(L_i)$ denote the size of all items in layer $i$. 
	Assume that $k \ge 2$. Then we have:
	\begin{enumerate}
		\item The cost of \FF arising from servers in $\mathcal{B}$ is $k(\ell+2).$
		\item $x(L_0) > k/2$;
		\item $x(L_i) + x(L_{i-1})/2 > (k-1)/2$ for $i \in \{1, \ldots, \ell\}$
	\end{enumerate}
	
\end{lemma}
\begin{proof}
	
	\begin{enumerate}
		
		\item Immediate from the definition of the cost and the fact that each server in $\mathcal{B}$ has duration exactly $\ell+2$. 
		\item This is a standard argument about \FF. Let $\gamma_i = x(L_0 \cap B_i)$ denote the total size of items in server $B_i$ that arrived at time $0$. We have $\gamma_i + \gamma_{i-1} > 1$ for $i \in \{2, 3, \ldots, k\}$ since otherwise items in server $B_i$ would have been placed in server $B_{i-1}$ instead.
		Adding up all these inequalities we get:
		
		\[ \sum_{i=2}^k (\gamma_{i-1} + \gamma_i) > k-1.\]
		
		By adding $\gamma_1$ and $\gamma_k$ to both sides, we obtain 
		$ 2 \sum_{i=1}^k \gamma_i > k-1 + \gamma_1 + \gamma_k.$
		Observe that $\gamma_1 + \gamma_k > 1$ by the same reasoning as before. Moreover, we have $\sum_{i=1}^k \gamma_i = x(L_0)$ by definition. Combining all these facts, establishes this part of the lemma.
		\item Fix $i \in \{1, \ldots, \ell\}$. Let $\gamma_j = x(L_{i-1} \cap B_j)$ denote the total size of items in server $B_j$ that arrived at time $i-1$. Similarly, let $\delta_j = x(L_i \cap B_j)$ denote the total size of items in server $B_j$ that arrived at time $i$. For $j \in \{2, 3, \ldots, k\}$ we have:
		
		\[ \delta_j +\gamma_{j-1} + \delta_{j-1} > 1\]

		otherwise items packed into bin $B_j$ at time $i$ should have been placed into bin $j-1$ by \FF. Summing all these inequalities we obtain:
		\[ \sum_{j=2}^k (\delta_j + \delta_{j-1} + \gamma_{j-1}) > k-1\]
		
		By adding $\delta_1 + \delta_k + \gamma_k$ to both sides we obtain:
		\[ 2 \sum_{j=1}^k \delta_j + \sum_{j=1}^k \gamma_j > k-1 + \delta_1 + \delta_k + \gamma_k.\]
		We have $\sum_{j=1}^k \delta_j = x(L_i)$ and $\sum_{j=1}^k \gamma_j = x(L_{i-1})$. Therefore, we obtain:
		
		\[ 2 x(L_i) + x(L_{i-1}) > k-1 + \delta_1 + \delta_k + \gamma_k.\]
	\end{enumerate}
\end{proof}
Now from the results in Lemma~\ref{lemma:observations}, we can establish strong upper bounds on the ratio between utilization and $\FF$ cost, as follows.
\begin{lemma}\label{lem:uniform_long_servers} Let $\sigma$ be the input such that each job has duration $2$ and arrival times $0, 1, 2, \ldots, \ell$ and suppose \FF opens $k \ge 2$ servers and each server of \FF on $\sigma$ starts at time $0$ and finishes at time $\ell + 2$. Then we have:
	\[ \frac{util(\sigma)}{\FF(\sigma)} > \frac{2}{3} - \frac{2}{3k} - \frac{2}{3(\ell+2)}.\]
where $util(\sigma)$ is the utilization of the input sequence which is defined as the total volume (size times duration) of all jobs.	
	
\end{lemma}
\begin{proof}
	Fix $\ell \ge 2$. From the results in Lemma~\ref{lemma:observations}, we have:\\
	\begin{align*}
	IN_i: x(L_{\ell-i}) + x(L_{\ell-i-1})/2 > (k-1)/2:&   \hspace{1cm} i\in \{0, 1, \ldots, \ell-1\} \\
	IN_\ell: x(L_0) > k/2 > (k-1)/2: &   \hspace{1cm} o.w\\
	\end{align*}
	
We will choose multipliers $f_0, f_1, \ldots, f_\ell$ such that the linear combination of inequalities $\sum_{i = 0}^{\ell} f_i IN_i$ has $util(\sigma)/2$ 
on the left hand side. 
The multipliers $f_i$ are defined recursively:
	
	\[f_i = \left\{ \begin{array}{ll}
	1 & \text{ if } i = 0,\\
	1-f_{i-1}/2 & \text{ if } i \ge 1.\\
	\end{array} \right.
	\]
 
Observe that the multipliers satisfy $f_i + \frac{1}{2} f_{i-1} = 1$ for $i \ge 1$. Next, we verify that $\sum_{i=0}^\ell f_i IN_i$ has $\sum_{i=0}^\ell x(L_i)$ on the left-hand side. For that it is convenient to think of inequality $IN_\ell$ as $x(L_0) + x(L_{-1})/2 > k/2$, where we define $x(L_{-1})/2 = 0$.
\begin{align*}
    \sum_{i = 0}^\ell f_i(x(L_{\ell-i}) + x(L_{\ell-i-1})/2) &= \sum_{i=0}^\ell f_i x(L_{\ell-i}) + \frac{1}{2} \sum_{i = 0}^\ell f_i x(L_{\ell-i-1})\\
    &=  f_0 x(L_\ell) + \sum_{i=1}^\ell f_i x(L_{\ell-i}) + \frac{1}{2} \sum_{i = 1}^\ell f_{i-1} x(L_{\ell-i})\\
    &= x(L_\ell) + \sum_{i=1}^\ell x(L_{\ell-i})\left(f_i + f_{i-1}/2 \right) = util(\sigma)/2,
\end{align*}
where the last equality is because the sum of sizes is half of utilization, since all items have duration $2$.

Let $F = \sum_{i = 0}^\ell f_i$. To compute a bound on $F$ observe the following: 
\begin{align*}
    \ell+1 &= f_0 + \sum_{i = 1}^\ell (f_i + f_{i-1}/2)  = \sum_{i = 0}^\ell f_i + \frac{1}{2} \sum_{i = 0}^{\ell-1} f_i \\
    &= F + \frac{1}{2} (F - f_\ell) = \frac{3}{2} F - \frac{1}{2} f_\ell.
\end{align*}
Thus, we can conclude that $F \ge \frac{2}{3} (\ell+1).$ This implies that the right-hand side of $\sum_{i=0}^\ell f_i IN_i$ is  $\frac{2}{3}F (k-1)/2 \geq \frac{2}{3}(\ell+1)(k-1)/2$. Combining with the calculation of the left-hand side, we obtain:
\[ util(\sigma) \ge \frac{2}{3} (k-1)(\ell+1).\]
 Since $\FF(\sigma) = k(\ell+2)$,  we conclude:
	 \[ \frac{util(\sigma)}{\FF(\sigma)} > \frac{2}{3} \frac{(k-1)(\ell+1)}{k(\ell+2)} = \frac{2}{3} \left( 1 - \frac{1}{k} \right) \left( 1-\frac{1}{\ell+2}\right) > \frac{2}{3} - \frac{2}{3k} - \frac{2}{3(\ell+2)}.\]

\end{proof}
\begin{corollary}\label{corollary:uniform_long_servers}
	On the inputs described in Lemma~\ref{lem:uniform_long_servers} with $k,\ell \rightarrow \infty$ \FF achieves competitive ratio $3/2$.
\end{corollary}
\begin{proof}
	The upper bound on the competitive ratio follows from Lemma~\ref{lem:uniform_long_servers}, since we have $OPT(\sigma) \ge util(\sigma)$ and $util(\sigma)/\FF(\sigma) = 2/3 + o(1)$.
\end{proof}
The bound on utilization in Lemma~\ref{lem:uniform_long_servers} is tight asymptotically as $k, \ell \rightarrow \infty$. It is witnessed by the following example. Fix $k, \ell \in \mathbb{N}$ and $\epsilon = 1/(k \ell)$.
\begin{itemize}
	\item at time $0$ we present $k$ items of size $2/3$ each;
	\item at odd times $1, 3, 5, \ldots$ we present $k$ items of size $1/3$ each;
	\item at even times $2, 4, 6, \ldots$ we present $k$ items of size $1/3 +\epsilon$ each.
\end{itemize}
Arrival time of the last job is $\ell$ in the above. Assume $\ell$ is even for simplicity of the presentation.
As can be seen in Figure~\ref{fig:Layers-Fig}, \FF opens $k$ servers of duration $\ell+2$ each and utilization is $2 (2/3) k + 2 (1/3) k (\ell/2) + 2 (1/3 + \epsilon) k (\ell/2) = (2/3) k (\ell+2) + \epsilon k \ell = (2/3) k (\ell+2) + 1$. Thus, for this example we have:
\begin{align*}
\frac{util(\sigma)}{\FF(\sigma)} &= \frac{(2/3) k (\ell+2) + 1}{k(\ell+2)} 
= \frac{2}{3}  + \frac{1}{k(\ell+2)} = \frac{2}{3} + o(1).
\end{align*}
\[ \]

\begin{figure}[H]
	\begin{center}
		\includegraphics[scale=0.7]{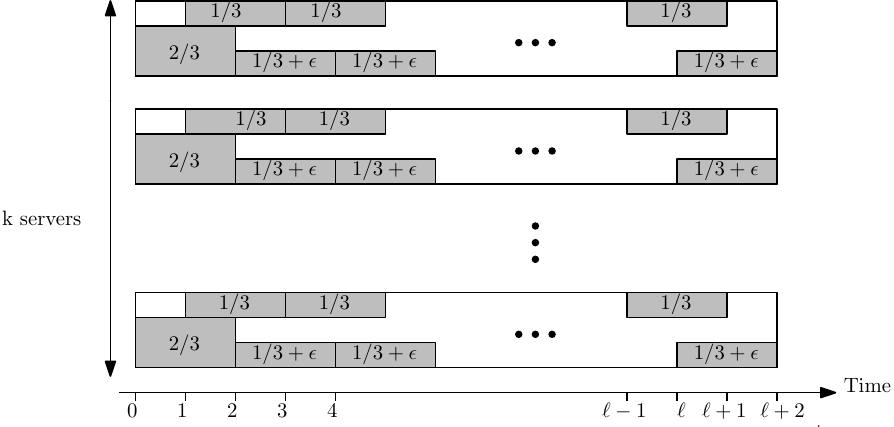}
		\caption{The bound in Corollary~\ref{corollary:uniform_long_servers} is shown in a tight example.}
		\label{fig:Layers-Fig}
	\end{center}
\end{figure}

\section{Conclusion}~\label{Section:Conclusion}
In this paper, we investigated the \RSiC problem to assign jobs to servers while minimizing the total cost of all rented servers. We considered the case where jobs have equal duration. We derived new bounds on the competitive ratios of two commonly used and well-studied algorithms, \NF and \FF. We proved that \NF is $2$-competitive. In the case that the arrival time of jobs belong to ${0, t}$ when $t \in (0, 1)$, the asymptotic competitive ratio of \FF is at least $\frac{34}{27} (t+1)$. This result shows that \FF has ratio $1.\overline{8}$ when $t = 1/2$ and has ratio $2.\overline{518}$ when $t \rightarrow 1$. Furthermore, we applied the weight function technique to the \RSiC problem for the first time. By using this technique, we established the upper bound on the competitive ratio of $\frac{168}{131}(1+t)$ for the case of jobs having arrival ${0, t}$ where $t > 0.559$. When $t \rightarrow 1$, this gives an upper bound of $2.559$ on the asymptotic competitive ratio of \FF for this case. 
Besides that, we obtained a strict competitive ratio of $2$ for \FF, where each job has a duration of $2$ and arrival time of $0$ or $1$.  Finally, it would be interesting to investigate the general case of \RSiC and attempt to close the gap between the upper and lower bounds for this problem.

\section{Acknowledgments}
This research is supported by the Natural sciences and Engineering Research Council of Canada (N.S.E.R.C). We would also like to thank the anonymous referees who provided detailed comments that helped us prepare a significantly improved version of the paper.
\bibliographystyle{abbrv}
\bibliography{references1}

\begin{thebibliography}{10}

\bibitem{balogh2017new}
J.~Balogh, J.~B{\'e}k{\'e}si, G.~D{\'o}sa, L.~Epstein, and A.~Levin.
\newblock A new and improved algorithm for online bin packing.
\newblock {\em arXiv preprint arXiv:1707.01728}, 2017.

\bibitem{chan2008dynamic}
J.~W.-T. Chan, T.-W. Lam, and P.~W. Wong.
\newblock Dynamic bin packing of unit fractions items.
\newblock {\em Theoretical Computer Science}, 409(3):521--529, 2008.

\bibitem{coffman1983dynamic}
E.~G. Coffman, Jr, M.~R. Garey, and D.~S. Johnson.
\newblock Dynamic bin packing.
\newblock {\em SIAM Journal on Computing}, 12(2):227--258, 1983.

\bibitem{dean2008mapreduce}
J.~Dean and S.~Ghemawat.
\newblock Mapreduce: simplified data processing on large clusters.
\newblock {\em Communications of the ACM}, 51(1):107--113, 2008.

\bibitem{dosa2013first}
G.~D{\'o}sa and J.~Sgall.
\newblock First fit bin packing: A tight analysis.
\newblock In {\em 30th International Symposium on Theoretical Aspects of
  Computer Science (STACS 2013)}. Schloss Dagstuhl-Leibniz-Zentrum fuer
  Informatik, 2013.

\bibitem{garey1976resource}
M.~R. Garey, R.~L. Graham, D.~S. Johnson, and A.~C.-C. Yao.
\newblock Resource constrained scheduling as generalized bin packing.
\newblock {\em Journal of Combinatorial Theory, Series A}, 21(3):257--298,
  1976.

\bibitem{garey1972worst}
M.~R. Garey, R.~L. Graham, and J.~D. Ullman.
\newblock Worst-case analysis of memory allocation algorithms.
\newblock In {\em Proceedings of the fourth annual ACM symposium on Theory of
  computing}, pages 143--150, 1972.

\bibitem{garey1979computers}
R.~Garey~Michael and D.~S. Johnson.
\newblock Computers and intractability.
\newblock {\em A guide to the theory of NP-completeness. A Series of Books in
  the Mathematical Sciences}, 1979.

\bibitem{han2010dynamic}
X.~Han, C.~Peng, D.~Ye, D.~Zhang, and Y.~Lan.
\newblock Dynamic bin packing with unit fraction items revisited.
\newblock {\em Information Processing Letters}, 110(23):1049--1054, 2010.

\bibitem{johnson1973near}
D.~S. Johnson.
\newblock {\em Near-optimal bin packing algorithms}.
\newblock PhD thesis, Massachusetts Institute of Technology, 1973.

\bibitem{johnson1974worst}
D.~S. Johnson, A.~Demers, J.~D. Ullman, M.~R. Garey, and R.~L. Graham.
\newblock Worst-case performance bounds for simple one-dimensional packing
  algorithms.
\newblock {\em SIAM Journal on computing}, 3(4):299--325, 1974.

\bibitem{kamali2015efficient}
S.~Kamali and A.~L{\'o}pez-Ortiz.
\newblock Efficient online strategies for renting servers in the cloud.
\newblock In {\em International Conference on Current Trends in Theory and
  Practice of Informatics}, pages 277--288. Springer, 2015.

\bibitem{lee1985simple}
C.~C. Lee and D.-T. Lee.
\newblock A simple on-line bin-packing algorithm.
\newblock {\em Journal of the ACM (JACM)}, 32(3):562--572, 1985.

\bibitem{li2014dynamic}
Y.~Li, X.~Tang, and W.~Cai.
\newblock On dynamic bin packing for resource allocation in the cloud.
\newblock In {\em Proceedings of the 26th ACM Symposium on Parallelism in
  Algorithms and Architectures}, pages 2--11, 2014.

\bibitem{ren2018combinatorial}
R.~Ren.
\newblock {\em Combinatorial algorithms for scheduling jobs to minimize server
  usage time}.
\newblock PhD thesis, 2018.

\bibitem{TangLRC2016}
X.~Tang, Y.~Li, R.~Ren, and W.~Cai.
\newblock On first fit bin packing for online cloud server allocation.
\newblock In {\em 2016 IEEE International Parallel and Distributed Processing
  Symposium (IPDPS)}, pages 323--332, 2016.

\bibitem{princeton1971performance}
J.~Ullman.
\newblock The performance of a memory allocation algorithm.
\newblock Technical report, Princeton University. Department of Electrical
  Engineering. Computer Science Laboratory, 1971.

\bibitem{wong20128}
P.~W. Wong, F.~C. Yung, and M.~Burcea.
\newblock An 8/3 lower bound for online dynamic bin packing.
\newblock In {\em International Symposium on Algorithms and Computation}, pages
  44--53. Springer, 2012.

\bibitem{xia2010tighter}
B.~Xia and Z.~Tan.
\newblock Tighter bounds of the first fit algorithm for the bin-packing
  problem.
\newblock {\em Discrete Applied Mathematics}, 158(15):1668--1675, 2010.

\end{thebibliography}
\end{document}